\documentclass[11pt,letterpaper]{article}

\usepackage{style}
\usepackage{shortcuts}
\usepackage{lowco2}

\newcommand*\samethanks[1][\value{footnote}]{\footnotemark[#1]}
\newif\ifdraft%
\drafttrue

\title{Deterministic Edge Coloring with few Colors in \CONGEST}

\author{Joakim Blikstad  \textcircled{r}\footnote{The author ordering was randomized using \url{https://www.aeaweb.org/journals/policies/random-author-order/} generator. It is requested that citations of this work list the authors separated by \texttt{\textbackslash textcircled\{r\}} instead of commas.}\hspace{.5em} 
Yannic Maus\thanks{TU Graz, Austria. This research was funded in whole or in part by the Austrian Science Fund (FWF) \url{https://doi.org/10.55776/P36280} and \url{https://doi.org/10.55776/I6915}. For open access purposes, the author has applied a CC BY public copyright license to any author-accepted manuscript version arising from this submission.} \hspace{.5em} \textcircled{r} \hspace{.5em} Tijn de Vos\samethanks}
\date{}

\begin{document}

  \begin{titlepage}
    \maketitle
    \thispagestyle{empty}

    \begin{abstract}
As the main contribution of this work we present  deterministic edge coloring algorithms in the \CONGEST model. In particular, we present an algorithm that edge colors any $n$-node graph with maximum degree $\Delta$ with $(1+\varepsilon)\Delta+O(\sqrt{\log n})$ colors in $\tilde{O}(\log^{2.5} n+\log^2 \Delta \log n)$ rounds. This brings the upper bound polynomially close to the lower bound of $\Omega(\log n/\log\log n)$ rounds that also holds in the more powerful \LOCAL model [Chang, He, Li, Pettie, Uitto; SODA'18].  As long as $\Delta \geq c\sqrt{\log n}$ our algorithm uses fewer than $2\Delta-1$ colors and to the best of our knowledge is the first polylogarithmic-round \CONGEST algorithm achieving this for any range of $\Delta$.  

As a corollary we also improve the complexity of edge coloring with $2\Delta-1$ colors for all ranges of $\Delta$ to $\tilde{O}(\log^{2.5} n+\log^2 \Delta \log n)$. This improves upon  the previous  $O(\log^8 n)$-round algorithm from [Fischer, Ghaffari, Kuhn; FOCS'17].

Our approach builds on a refined analysis and extension of the online edge-coloring algorithm of Blikstad, Svensson, Vintan, and Wajc [FOCS’25], and more broadly on new connections between online and distributed graph algorithms. We show that their algorithm exhibits very low locality and, if it can additionally  have limited local access to future edges (as distributed algorithms can), it can be derandomized for smaller degrees. Under this additional power, we are able to bypass classical online lower bounds and translate the results to efficient distributed algorithms. This leads to our \CONGEST algorithm for $(1+\varepsilon)\Delta+O(\sqrt{\log n})$-edge coloring. Since the modified online algorithm can be implemented more efficiently in the \LOCAL model, we also obtain (marginally) improved complexity bounds in that model.

\end{abstract}

    \vfill
    \lowcotwo\ This is a low-co2 research paper: \lowcotwourl[\lowcotwoversion]. This research was developed, written, submitted and presented without the use of air travel.

    \clearpage
    \tableofcontents
    \thispagestyle{empty}
    \newpage
    \thispagestyle{empty}
      \listoftodos
  \end{titlepage}

\section{Introduction}
\label{sec:intro}
In the edge coloring problem, the goal is to assign a color to each edge of an input graph so that any two adjacent edges receive different colors. 
The problem has been extensively studied across a wide range of computational models, including centralized, sublinear, massively parallel, dynamic, and online settings~\cite{arjomandi1982efficient,BMN92,harvey2018greedy,behnezhad2019streaming,duan2019dynamic,bhattacharya2021online,ansari2022simple,christiansen2023power,bhattacharya2024faster,Assadi25,BlikstadSVW25,chechik2025improved}. 
This work focuses on designing better algorithms in the distributed setting, specifically within the \CONGEST model. In this model \cite{peleg00},  a communication network is abstracted as an $n$-node graph, where the vertices communicate with each other through the edges of the graph in synchronous rounds by sending a message of $O(\log n)$ bits to each neighbor. The complexity measure is the number of communication rounds  until every node has produced its output, e.g., has output the colors of its incident edges. 

\paragraph{Historic focus on \boldmath $(2\Delta-1)$-edge coloring within distributed computing.} Edge coloring is one of the central symmetry-breaking problems that has been extensively studied in both the \CONGEST model and also in the more powerful \LOCAL model. The latter is  identical to the \CONGEST model, except that messages can be of unbounded size~\cite{linial92}. Classically, research in the field has focused on the \emph{greedy version} of the problem in which one aims to use $2\Delta-1$ colors for graphs with maximum degree $\Delta$. For constant-degree graphs, $O(\log^* n)$-round deterministic algorithms have been known for decades~\cite{linial92,panconesi-rizzi}. For general graphs, a breakthrough result in 2017 provided the first polylogarithmic-time deterministic algorithm for the problem~\cite{FGK17}.

\paragraph{Coloring with fewer colors.}
The number of colors used by the aforementioned algorithms is far from optimal. A seminal result by Vizing shows that every graph admits an edge coloring with at most 
$\Delta+1$ colors~\cite{vizing1964estimate}. Consequently, across all of the models discussed above, substantial effort has been devoted to designing algorithms that approach this bound as closely as possible.

In the randomized distributed setting, algorithms achieving substantially fewer than $2\Delta-1$ colors are known. Already a single round in which each edge picks a tentative color and permanently adopts it if no adjacent edge picks the same color creates substantial slack for solving the problem and reducing the number of colors.  Hence, there are several randomized  algorithms usually targeting $(1+\eps)\Delta$ colors, starting with the relatively slow R\"odl nibble method-based algorithm from Panconesi, Dubashi, and Grable \cite{NibbleMethod}, a $(1+\eps)\Delta$-edge coloring for sufficiently large $\Delta$~\cite{EPS15}, and a  $\Delta+O(\log n/\log\log n)$-edge coloring algorithm  by Su and Vu in $\poly(\Delta,\log n)$ rounds~\cite{SuVu19}. Lately multiple works devised $\poly\log\log n$-round algorithms for coloring with $(1+\eps)\Delta$ colors \cite{Davies23,HMN22} and another work even obtained an $O(\log^*n)$-round algorithm for large degree graphs \cite{HN21}. While all of the former results are for the \LOCAL model, the latter two   hold in \CONGEST. 


Deterministically, the situation is more challenging. Fast distributed algorithms rely on many nodes deciding simultaneously, and here the \emph{greedy nature} of $2\Delta-1$ edge coloring is crucial: any partial coloring can always be extended to a complete one, which enables highly parallel progress. With fewer colors this guarantee disappears: partial colorings may block further progress and fail to extend to a valid solution. This obstruction is inherent: \cite{CHLPU18} shows that any deterministic distributed algorithm requires $\Omega(\log_{\Delta} n)$ rounds, ruling out sublogarithmic-time algorithms.

In the \LOCAL model, edge colorings with few colors can nevertheless be obtained either via intricate procedures, e.g., \cite{GKMU18,BMNSU25,JMS25}, or, as the state of the art for general graphs, is obtained through applying the powerful general derandomization techniques \cite{GhaffariKM17,GHK18,RG20,GG24} to the aforementioned randomized results; all these algorithms run in polylogarithmic time. These approaches, however, do not extend to \CONGEST, where no general derandomization method is known. Despite efficient randomized \CONGEST algorithms achieving fewer than $(1+\varepsilon)\Delta$ colors, no deterministic counterpart is known and the aforementioned \LOCAL model derandomization results all rely on gathering large neighborhoods and apply brute force locally which is infeasible in \CONGEST. This motivates the following question.

\begin{center}
\vspace{-0.5em}
\emph{Is there a deterministic \CONGEST algorithm that edge colors graphs of maximum degree $\Delta$ with fewer than $2\Delta-1$ colors in polylogarithmic time?}
\end{center}


\subsection{Our Contributions}

\paragraph{Contribution 1: Edge coloring  in \boldmath \CONGEST. }

As our main contribution we answer the question in the affirmative unless the maximum degree $\Delta$ is too small. 
We prove the following theorem, where we use $\tilde O(f) = O(f\poly \log f)$.

\begin{restatable}{theorem}{thmCONGESTmain}
\label{thm:CONGESTmain}
For any constant $\eps>0$, there is a deterministic \CONGEST algorithm to edge color any $n$-node graph with maximum degree $\Delta$ with $(1+\eps)\Delta+O(\sqrt{\log n})$ colors in 
$\tilde O(\log^{2.5} n+\log^2\Delta\log n)$
rounds. 
\end{restatable}
In fact, \Cref{thm:CONGESTmain} produces a coloring with less than $2\Delta-1$ colors when $\Delta\ge c\sqrt{\log n}$ for a sufficiently large constant $c$. 
To the best of our knowledge there is no prior \CONGEST algorithm for deterministic edge coloring with fewer than $2\Delta-1$ colors, regardless of the choice of $\Delta$. 
For $\Delta\ge \sqrt{\log n}$, the aforementioned lower bound $\Omega(\log_{\Delta}n)$ for the problem amounts to $\Omega(\log n/\log\log n)$  \cite{CHLPU18}; \Cref{thm:CONGESTmain} comes polynomially close to this. 

As a corollary of \Cref{thm:CONGESTmain} we also obtain faster algorithms for coloring with $2\Delta-1$ colors. 
This improves on the $O(\log^8 n)$ round complexity claimed in~\cite{FGK17}.\footnote{The algorithm in \cite{FGK17} is formulated in the \LOCAL model but a footnote states that it can be implemented in \CONGEST; it is unclear whether this requires additional overhead due to congestion. }


\begin{restatable}{corollary}{CorTwoDelta}
\label{cor:CONGEST}
There is a deterministic \CONGEST algorithm to edge color any $n$-node graph with maximum degree $\Delta$ with $2\Delta-1$ colors in $\tilde O(\log^{2.5}n+\log^2 \Delta \log n)$ rounds.
\end{restatable}




\paragraph{Contribution 2: Improvements in the \boldmath \LOCAL model.}
Our techniques also yield a slightly faster algorithm in the \LOCAL model compared to the previous state of the art. \Cref{thm:LOCALmainNoND} is the \LOCAL-model counterpart of \Cref{thm:CONGESTmain} and achieves an improved runtime, as the absence of bandwidth restrictions eliminates the need to handle congestion. 

\begin{restatable}{theorem}{thmLOCALmainNoND}
\label{thm:LOCALmainNoND}
For any constant $\eps>0$, there is a deterministic \LOCAL algorithm to edge color any $n$-node graph with maximum degree $\Delta$ with $(1+\eps)\Delta+O(\sqrt{\log n})$ colors in $O(\log^{2}n)+\tilde O(\log^2\Delta\log n)$ rounds.
\end{restatable}
For $\Delta= 2^{O(\sqrt{\log n})}$, our running time becomes $O(\log^2 n)$, removing the $\poly \log \log n$ factor from the state of the art~\cite{CHLPU18,Davies23,RG20,HN21,GG24}. In this introduction, we restrict ourselves to the case where $\eps$ is constant. However we note that some of the earlier algorithms were able to achieve even fewer colors, e.g., $\Delta+o(\Delta)$ colors in $\tilde O(\log^2n)$ rounds as they allow for setting $\eps=o(1)$, obtained by combining~\cite{Davies23} with black-box derandomization~\cite{GhaffariKM17,GHK18,RG20,GG24}.

\paragraph{The additive \boldmath $O(\sqrt{\log n})$ term in the number of colors.}
In our technical sections we design algorithms that color with $(1+\eps)\Delta$ colors (without the additive $O(\sqrt{\log n})$ term) under the assumption that $\Delta\geq c\sqrt{\log n}$ holds for a sufficiently large constant $c$. \Cref{thm:CONGESTmain,thm:LOCALmainNoND,cor:CONGEST}  follow from this result as for $\Delta=O(\sqrt{\log n})$, one can compute an edge coloring with $2\Delta-1\leq (1+\eps)\Delta+O(\sqrt{\log n})$ colors in $O(\Delta+\log^*n)=O(\log^{o(1)} n)$ rounds in \CONGEST (and hence in \LOCAL)~\cite{BEG17}.

The fact that our deterministic $(1+\eps)\Delta$-edge coloring  algorithm works for maximum degrees as low as $\Omega(\sqrt{\log n})$ is significant, as $\Theta(\log n)$ has not only been a provable barrier for deterministic online algorithms\footnote{Online algorithms play an important role in our techniques.} \cite{BMN92} , and a natural barrier for the deterministic algorithm of \cite{GHKM18}, but has also appeared as a limitation in randomized settings, for example, in the near-linear-time randomized centralized algorithm by Assadi where this limitation is rooted in Chernoff-type concentration bounds \cite{Assadi25}.

Moreover, $\Omega(\log n)$ has also been a barrier in the distributed setting: More precisely,  restricted to the case $\Delta=\Omega(\log n)$ the $(1+\eps)\Delta$-edge coloring algorithms of \cite{NibbleMethod,CHLPU18} (combined with the algorithm of \cite{HN21} for greedy edge coloring) run in $O(\log^*n)$ rounds.  To obtain $(1+\eps)\Delta$-edge coloring algorithm that work for smaller values of $\Delta$ these works fall back to the Lovász Local Lemma which causes runtime increase in the \LOCAL model, see, e.g.,  \cite{HMN22,Davies23}, and is poorly understood in the \CONGEST model~\cite{MU21}.

In the remainder of the paper, we focus on computing a $(1+\eps)\Delta$-edge coloring for $\Delta \geq c\sqrt{\log n}$.

\paragraph{Contribution 3: Connection between online and distributed algorithms and a succinct overview of our techniques.} 
Our technique builds on analyzing and extending an existing online algorithm, and lifting it to distributed models of computation. In a recent work Blikstad, Svensson, Vintan, and Wajc provide  deterministic (randomized) online $(1+\eps)\Delta$ edge coloring algorithms that work for $\Delta= \Omega(\log n)$ and $\Delta\geq \Omega(\sqrt{\log n})$, respectively \cite{BlikstadSVW25}. 
We build on their algorithm and uncover new connections to the area of distributed graph algorithms. 
In the online setting, edges arrive in an adversarial sequence and the algorithm must irrevocably assign a color upon arrival~\cite{BMN92}. While such decisions may depend on the entire past, translating these algorithms to the distributed setting requires that each decision is based only on local information around the arriving edge.
We show that the algorithm of \cite{BlikstadSVW25} indeed has low \emph{locality}. Furthermore, we show that if the algorithm's deterministic choices are additionally allowed to access limited information about future edges within a small neighborhood, it remains effective even for degrees as low as $\Omega(\sqrt{\log n})$ thereby circumventing the logarithmic lower bound on $\Delta$ for deterministic online edge coloring  from \cite{BMN92}. Although online algorithms cannot see future edges, distributed algorithms  can see all edges in the vicinity.

\subparagraph{Towards efficient distributed \CONGEST algorithms.}
As a result of the above we obtain that $(1+\eps)\Delta$-edge coloring is a local greedy problem unless the maximum degree $\Delta$ is less than $\Theta(\sqrt{\log n})$. The \emph{greedy} steps, however, require carefully designed color choices that guarantee the algorithm never gets stuck in the future. 
Such a result can also be derived from \cite{NibbleMethod,CHLPU18} combined with black-box derandomization~\cite{GHK18} but only for $\Delta=\Omega(\log n)$.

A central benefit of our greedy algorithm with local greedy choices is that it can be implemented in the \CONGEST model with moderate overhead by iterating through an appropriate schedule of edges and dealing with congestion. Optimizing the schedule for the internals of the greedy choices of the pseudo-online algorithm and optimizing the implementation for congestion   yields a \CONGEST algorithm with complexity $\poly\Delta\cdot O(\log^* n)$ --oversimplified, it becomes $\tilde O(\Delta^5\cdot \log^* n)$ rounds-- for using $(1+\eps)\Delta+O(\sqrt{\log n})$ colors. Using degree-splitting we can split into smaller instances in a divide-and-conquer fashion using $\tilde{O}(\log^2\Delta\log n)$ rounds, see \Cref{sec:congest,sec:Congest_DS} for details. As a result, we can restrict to the case where $\Delta=\Theta(\sqrt{\log n})$ yielding the $\tilde{O}(\log^{2.5} n)$-round algorithm of \Cref{thm:CONGESTmain}. Besides additional potential overhead of derandomizing the algorithms of \cite{NibbleMethod,CHLPU18}, decreasing the threshold where $(1+\eps)\Delta$ becomes a greedy problem to $\Theta(\sqrt{\log n})$ directly improves the runtime \Cref{thm:CONGESTmain}.

\paragraph{Outline.} 
Further related work appears in \Cref{sec:related work}.
 In \Cref{sec:technical_overview}, we formalize the relationship between online, greedy, and distributed algorithms discussed above.
 Then, in \Cref{sec:online_locality} we show that $(1+\eps)\Delta$-edge coloring has low locality, with corresponding lower bounds provided in \Cref{sec:LB}. We provide our \LOCAL algorithms in \Cref{sec:local}, and finally our \CONGEST algorithms in \Cref{sec:congest}.


\section{Background and Extensive Technical Overview}\label{sec:technical_overview}

While our core contribution lies in a fast, polylogarithmic-time \CONGEST algorithm for edge coloring, our techniques are significantly easier to grasp through the lens of the more powerful \LOCAL model. Only at the very end of this section, i.e., in \Cref{sec:tecOverviewCongest}, we provide the challenges with implementing our algorithm in the \CONGEST model.



\paragraph{Outline of \Cref{sec:technical_overview}.}
In \Cref{sec:tecOverviewSLOCAL} we formally define greedy algorithms with low locality, in \Cref{sec:tecOverviewRelationOnline} we relate this definition to online algorithms, and in \Cref{sec:tecOverViewOnline} we present the online algorithm of \cite{BlikstadSVW25} as well as an overview of our new analysis. In \Cref{sec:tec:scheduling} we show how these results imply new distributed algorithms in the \LOCAL model. Finally, in \Cref{sec:tecOverviewCongest}, we present the challenges with implementing the algorithm in \CONGEST and how we overcome them.

\subsection{Sequential Greedy Algorithms with Low Locality} \label{sec:overview_greedy}
\label{sec:tecOverviewSLOCAL}
The central model capturing whether a problem is a greedy algorithm with small locality is the \SLOCAL model introduced in \cite{GhaffariKM17}. 

\paragraph{\boldmath \SLOCAL model.} In the \SLOCAL model, the nodes are processed sequentially in some arbitrary order chosen by an adversary. When node $v \in V$ is processed, it chooses its internal state and output based on the information available within distance $T(n)$ in the graph. This information also includes the internal states of the nodes processed before node $v$. The radius $T(n)$ is referred to as the \emph{locality} of the algorithm.
%
We note that any \LOCAL algorithm with locality $T(n)$ is also an \SLOCAL algorithm with locality $T(n)$. In fact, \SLOCAL is a strictly stronger model. 

The model was introduced in \cite{GhaffariKM17} as a means to gain insights into distributed algorithms and complexity theory and most importantly to understand how one could obtain efficient deterministic, i.e., $\poly\log n$-round, algorithms for greedy-type problems like maximal independent set, or $\Delta+1$-vertex coloring. In a complexity theoretic approach, the work showed that several problems are complete in the sense that a $\poly\log n$-round algorithm for any complete problem would imply a $\poly\log n$-round algorithm for any problem that can be solved with polylogarithmic locality in \SLOCAL. 
Later, the \SLOCAL model was also used for a blackbox derandomization result\footnote{In fact, for any problem for which solutions can be verified efficiently with a deterministic  \LOCAL model an efficient randomized \LOCAL algorithm also implies an efficient deterministic \SLOCAL algorithm. } \cite{GHK18}, until in 2020 Rozho\v{n} and Ghaffari \cite{RG20} designed the first $\poly\log n$-round deterministic algorithm for one of the complete problems, namely for the network decomposition problem, yielding deterministic polylogarithmic-time \LOCAL algorithms for most problems in the field. 

However, the larger the complexity of a problem is in the \SLOCAL model, the larger is the multiplicative overhead in these generic reductions. 
While the original work introducing the model did not care for polylogarithmic factors, neither in the \LOCAL model nor in the \SLOCAL model, follow-up work did, e.g.  \cite{GS17,GHK18,BalliuKKLOPPR0S23} provided tight complexities for the sinkless orientation problem in both models.
We aim to contribute to a more fine-grained understanding of complexities in the \SLOCAL model. This motivates the following definition. 


\begin{definition*}
    A problem is a \emph{(low locality) greedy-type problem} if it can be solved with constant locality in \SLOCAL.
\end{definition*}
Examples of greedy-type problems are maximal independent set, maximal matching, $\Delta+1$-vertex coloring and $2\Delta-1$ edge coloring. 
\cite{NibbleMethod,CHLPU18} imply that  $(1+\eps)\Delta$-edge coloring is a greedy-type problem for $\Delta = \Omega(\log n)$. We extend this to the regime $\Delta = \Omega(\sqrt{\log n})$. As explained earlier, this extension is also essential to get our efficient \CONGEST algorithm. 
\begin{restatable}{theorem}{thmSLOCALmain}
\label{thm:SLOCALmain}
For any constant $\eps>0$, there is a deterministic algorithm that edge colors any $n$-node graph with maximum degree $\Delta\geq c\sqrt{\log n}$ with $(1+\eps)\Delta$ colors in the \SLOCAL model with constant locality, for some sufficiently large constant $c>0$. 
\end{restatable}


In the following, we show how \Cref{thm:SLOCALmain} follows via an \emph{online algorithm}. We note that our general approach for derandomization to obtain \Cref{thm:CONGESTmain,thm:LOCALmainNoND} only requires an algorithm to have low locality. Hence similar results can for example be obtained via \cite{NibbleMethod,CHLPU18} for the regime $\Delta =\Omega(\log n)$. As explained earlier, this leads to a higher $\poly \log n$ round complexity in the \CONGEST model.

\subsection{Relation between Sequential Local Algorithms and Online Algorithms}
\label{sec:tecOverviewRelationOnline}


 An online algorithm is an algorithm that processes its input piece by piece in a sequential manner, making decisions irrevocably at each step without knowledge of future input. The \SLOCAL model has a natural relation to online algorithms; vertices arrive one after the other, and the online/\SLOCAL algorithm makes irrevocable decisions on the output of a node. The central difference between an online algorithm and an \SLOCAL algorithm is that an online algorithm can base its decision on \emph{all} information on nodes that it has already processed and cannot see future inputs, while in \SLOCAL, a vertex is restricted to only read the information stored in the $T(n)$-hop radius around it, but it also has access to all vertices in this radius to come.

 \begin{observation*}
     Any online algorithm whose decisions on the output of a node $v$ only depend on information stored within $T(n)$ hops from $v$ is an \SLOCAL algorithm with locality $T(n)$.
 \end{observation*}

 If the online algorithm is deterministic, then so is the \SLOCAL algorithm. If the online algorithm holds w.h.p.\ against an oblivious adversary, it holds w.h.p.\ in the \SLOCAL model. The reason for the latter is that the fixed graph arrives in an arbitrary order in the \SLOCAL model, which can be seen as an oblivious adversary in the online model.
 
While existing online algorithms operate under edge arrivals, the \SLOCAL model is based on \emph{vertex arrivals}. To achieve the precise bounds in \Cref{thm:LOCALmainNoND,thm:CONGESTmain}, lifting the \SLOCAL-based \Cref{thm:SLOCALmain} to the \CONGEST and \LOCAL model respectively is insufficient, as it is lossy in the locality radius. To this end and also since edge arrivals are intrinsic to online edge coloring, we introduce an edge-based variant of the \SLOCAL model; it is similar to the \SLOCAL model on the line graph of the input graph.


\paragraph{\boldmath Edge \SLOCAL.}
Formally, we define the \SLOCALedge model as follows. The edges are processed sequentially in some arbitrary order chosen by an adversary. When edge $e \in E$ is processed, it chooses its internal state and output based on the information available within distance $T(n)$ in the input graph $G$. This information includes edges arriving in the future but it also includes the internal states of the edges processed before edge $e$. For distances between edges, we say that $d_G(e,f):=d_{L(G)}(e,f)$, where $L(G)$ is the \emph{line graph} of $G$: $L(G):=(E,\{\{uv,vw\}:uv,vw\in E\})$, i.e., it has the edges of $E$ as vertices and two edges are connected if and only if they share a vertex in $G$.
In other words, it counts how many vertices are between $e$ and $f$, see \Cref{fig:path} for an example. 

\begin{figure}[!ht]
    \includegraphics[page=1]{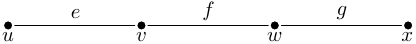}
    \centering
    \caption{In this graph $d(e,f)=1$ and $d(e,g)=2$.}
    \label{fig:path}
    \vspace{-0.7em}
\end{figure}

One can choose here, whether with locality $T(n)$, we do or do not know the degrees of the endpoints of the last edge. For example, if in \Cref{fig:path} we have locality $1$, does $e$ know the degree of $w$?
In this paper, the weaker version suffices -- where we do not know the vertex degrees. If we do know the vertex degrees, then the \SLOCALedge model is equivalent to the \SLOCAL on the line graph of the input graph. 

\begin{observation*}
     Any edge-arrival online algorithm that bases its output for an edge $e$ only on information within $T(n)$ hops from $e$ is an \SLOCALedge algorithm with locality $T(n)$.
\end{observation*}
Below, we explain how an \SLOCALedge algorithm can be  run in the distributed setting at the cost of a suitable parallel schedule of edges. As our main technical contribution, we show the following theorem about the \SLOCALedge complexity of edge coloring. For a parameter $t$ and some computational model $\mathsf{MODEL}$ the complexity class $\mathsf{MODEL}(t)$ denotes the class of problems that have complexity at most $t$ in the model. 
\begin{restatable}{theorem}{ThmOnlineisESLOCAL}\label{thm:online_is_ESLOCAL}
    For any constant $\eps>0$, there is a randomized \SLOCALedge{}(1) algorithm and a deterministic \SLOCALedge{}(5) algorithm to color any $n$-node graph with maximum degree $\Delta\ge c\sqrt{\log n}$ with $(1+\eps)\Delta$ colors, for some sufficiently large constant $c>0$.
\end{restatable}
\Cref{thm:online_is_ESLOCAL} immediately implies that the \SLOCAL locality of the respective edge coloring problem is also constant, proving \Cref{thm:SLOCALmain}. 

\paragraph{Lower bounds in \boldmath \SLOCALedge.}
Before we go into the technical details of \Cref{thm:online_is_ESLOCAL}, let us state the next theorem lower bounding the required locality for edge coloring in \SLOCALedge.
\begin{restatable}{theorem}{ESLOCALLB}\label{thm:LB_det}
    Any deterministic edge coloring algorithm in \SLOCALedge{} of locality $\le 2$ needs
    to use $2\Delta-1$ colors, even on graphs of degree $\Delta \ge c\sqrt{n}$, for some constant $c>0$.

    Any randomized edge coloring algorithm in \SLOCALedge{} of locality $\le 2$ needs
    to use $2\Delta-1$ colors on graphs of degree $\Delta = O(\sqrt{\log n})$.
\end{restatable}
\Cref{thm:LB_det} shows that the locality in \SLOCALedge of edge coloring with few colors has to be at least $3$, while our upper bound shows that locality $5$ suffices. 

\paragraph{Scheduling via conflict graphs.} Instead of translating the locality 5 online algorithm (or one that can also have local access to the future) in a blackbox manner, we further optimize the online algorithm to allow more edges to be processed simultaneously, even when they are close together.  We identify which edges can be handled in parallel by constructing a conflict graph. This conflict graph then enables us to compute a shorter schedule for parallelizing the execution of the online algorithm in the distributed setting. 
As a result we obtain a distributed algorithm using $O(\Delta^4+\log^* n)$ rounds. 
In combination with the explained divide-and-conquer approach we obtain  \Cref{thm:LOCALmainNoND}. 

\subsection{The Local Greedy Online Algorithm}
\paragraph{Online coloring.}
\label{sec:tecOverViewOnline}
In online edge coloring, the input is an unknown graph of maximum degree $\Delta$ (this value is known in advance). At each time $t = 1, \dots , m$, an edge $e_t$ is revealed. An online algorithm must decide for this edge on its arrival what color to assign it, immediately and irrevocably. The subgraphs induced by each color must form a matching; put otherwise, each node must have at most one edge of each color. 

\paragraph{Online edge coloring algorithm of \cite{BlikstadSVW25}.}
For this overview, we describe a simplified version of the randomized online algorithm of \cite{BlikstadSVW25} (see \cref{alg:edge-coloring} for the full algorithm). This simplified version works when $\Delta \ge \log n$, but with a more complicated algorithm~\cite{BlikstadSVW25} shows how one can make it work starting at $\Delta = c\sqrt{\log n}$ for a sufficiently large constant $c$; the latter algorithm is the one which we use in the body of our paper. Certain bad events happen in the execution of the randomized algorithm with probability $2^{-\Delta}$, and when $\Delta = \Omega(\log n)$ this is enough for ``with high probability'' concentration. By looking at larger neighborhoods, \cite{BlikstadSVW25} manages to bound the probability of bad events by $2^{-\Delta^2}$ instead, allowing for $\Delta$ being as low as $\approx \sqrt{\log n}$. This condition on $\Delta$ cannot be lowered further for \emph{any} randomized online algorithm: \cite{BMN92} show that any randomized online edge coloring algorithm needs to use $2\Delta-1$ colors whenever $\Delta = O(\sqrt{\log n})$.

\begin{tcolorbox}
Fix a palette $C = \{1,2,\dots,\Delta\}$.\\[0.05cm]
For each $e \in E$ and $c \in C$, initialize $P_{ec} = (1-\eps)/\Delta$.\\[0.05cm]
When an edge $e$ arrives:
\begin{itemize}
    \item \textbf{Step~1:} Sample color $c$ with probability $P_{ec}$; if no color is chosen (i.e., with probability $1 - \sum_{c} P_{ec}$), leave edge $e$ uncolored.
    \item \textbf{Step~2:} For each adjacent edge $f$ and color $c\in C$, update:
    \[
    P_{fc} = \begin{cases}
        0 & \text{if } c \text{ was assigned to } e,\\
        P_{fc} / (1 - P_{ec}) & \text{otherwise}.
    \end{cases}
    \]
    \item \textbf{Step~3:} If the edge was left uncolored, fall back to color it greedily with a separate palette $\{\Delta+1, \Delta+2, \ldots\}$.
\end{itemize}
\end{tcolorbox}

The algorithm maintains values $P_{ec}$ for each (future) edge $e$ and color $c\in [\Delta]$, initially set to $P_{ec} = \frac{1-\eps}{\Delta}$. This ensures that each edge has approximately probability $\sum_{c\in [\Delta]} P_{ec} = 1-\eps$ of receiving a color from the main palette; the updating rule $P_{fc}=0$ ensures that there are no conflicts for edges colored in Step~1. Edges that remain uncolored from the main palette $[\Delta]$ are forwarded to a standard greedy algorithm using a fresh palette with at most $2\Delta'-1$ additional colors, where $\Delta'$ is the maximum degree in the subgraph forwarded to greedy. Since each edge has probability $\approx 1-\eps$ of being colored from $[\Delta]$, we expect at most $\Delta' \approx \eps\Delta$ edges incident to each vertex to be forwarded to greedy. Consequently, the total number of colors used is at most $\Delta + 2\Delta'-1 = \Delta + O(\eps \Delta)$.  

When edge $e$ arrives, a color $c$ is sampled according to the probability distribution defined by weights $P_{ec_1}, P_{ec_2}, \ldots$, or the edge remains uncolored with probability $1-\sum_c P_{ec}$. The algorithm then updates the values $P_{fc}$ for neighboring edges $f$ accordingly. For the color $c$ assigned to $e$, we set $P_{fc} = 0$ since color $c$ cannot be reused at edge $f$. For all other colors not used at $e$, we scale up $P_{fc}$ by a factor of $\frac{1}{1-P_{ec}}$ to preserve the marginal probabilities that color $c$ will be used at the shared vertex of edges $e$ and $f$.

By inspecting the pseudo-code, we can observe that it can be implemented (randomized) with low locality. In \SLOCALedge{}, each arriving edge need only compute its $P_{ec}$ values by examining its already-arrived neighboring edges, thus yielding a randomized \SLOCALedge{} algorithm with locality $1$.

The analysis in \cite{BlikstadSVW24bipartite} critically depends on showing that $\sum_c P_{ec}$ remains well-concentrated around their initial expectation for future arriving edges. Specifically, two key properties must hold:
\begin{itemize}
\item $\sum_c P_{ec} < 1$, ensuring that the values form a valid probability distribution and the algorithm is well-defined.
\item $\sum_c P_{ec} > 1-2\eps$, guaranteeing that the probability of an edge being forwarded to greedy is at most $2\eps$.
\end{itemize}

In \cite{BlikstadSVW25}, this analysis is performed by constructing several martingales that control the required concentration bounds. We refer readers to \cite{BlikstadSVW25} for additional intuition behind the algorithm and a detailed analysis.

\paragraph{Derandomizing the algorithm.}
Derandomizing the algorithm while preserving low locality is far from trivial and most of \Cref{sec:online_locality} is spent to this end. To achieve this, we must carefully examine the randomized algorithm's analysis and ensure that both the algorithm and its underlying analysis remain local. Fortunately, the analysis in \cite{BlikstadSVW25} relies primarily on martingales and Azuma's inequality, where each martingale controls the probability of specific ``bad'' events such as $\sum_{c} P_{ec}$ growing too large or vertices receiving too many incident edges forwarded to greedy. 

Crucially, all martingales in this analysis are \emph{locally defined}, which enables the application of standard derandomization techniques such as conditional expectations or pessimistic estimators. We employ pessimistic estimators, which yield a potential function that, while somewhat complex to describe, ensures deterministic performance with identical guarantees to the randomized algorithm when the algorithm selects colors that minimize this potential at each step.

In essence, every application of Azuma's inequality or Chernoff bounds in \cite{BlikstadSVW25} translates to a \emph{local} potential function centered at some vertex, computable using only information within that vertex's $3$-hop neighborhood. This locality property allows our \SLOCALedge{} algorithm to identify all relevant vertices (within distance $2$) upon edge arrival, compute their potential functions (by examining the distance-$3$ neighborhood), and make locally optimal deterministic decisions that minimize potential changes. In total, this algorithm has \SLOCALedge locality $5$.

\begin{remark}
We note that \cite{BlikstadSVW25} also designs a \emph{deterministic} online algorithm, by black-box derandomizing a randomized online algorithm that works against an \emph{adaptive adversary}.  Such a black-box derandomization is always possible in the online setting, but not in the distributed setting as it has very high locality. 
Instead, we use the randomized algorithm of \cite{BlikstadSVW25} that works against the weaker \emph{oblivious adversary}. This is because that algorithm also works when $\Delta \approx \sqrt{\log n}$, while their deterministic (or randomized algorithm against an adaptive adversary) works only when $\Delta \ge \log n$.
In the online setting, it would be impossible to derandomize the algorithm against an oblivious adversary. This is because a lower bound by \cite{BMN92} that says that any deterministic online edge coloring algorithm needs to use $2\Delta-1$ colors whenever $\Delta = O( \log n)$.

Importantly, we can overcome this online lower-bound when working in the \SLOCALedge{} model. This is since the \SLOCALedge{} algorithm can, unlike a standard online algorithm, see the non-arrived edges. For example, if we have a path $a$--$b$--$c$--$d$, and the edge $(a,b)$ is arriving first. Then, when the edge $(c,d)$ arrives, in the \SLOCALedge{} model (if locality $\ge 2$) it will be able to see that it is distance 2 away from the edge $(a,b)$, and can use this fact in the decision-making (while an online algorithm will not know this).
\end{remark}

\subsection{Running Sequential Local Algorithms in \LOCAL}\label{sec:tec:scheduling}
Our distributed algorithms follow from the efficient parallelization of the deterministic \SLOCALedge algorithm with locality $5$ from \Cref{thm:online_is_ESLOCAL}, similar to the approach to run \SLOCAL algorithms in the \LOCAL model \cite{GhaffariKM17}. We obtain a generic translation of \SLOCALedge to the \LOCAL model. 

\vspace{-.5em}
\begin{restatable}{corollary}{CorESLOCAL}
\label{cor:complexity}
For any $\ell> 0$, we have that 
\begin{equation*}
    \text{\SLOCALedge{}}(\ell)\subseteq \LOCAL(O(\ell\cdot \Delta^\ell)+O(\log^* n)).
\end{equation*}
\end{restatable}

\paragraph{Proof idea of  \Cref{cor:complexity}.} From a \LOCAL model point of view, the purely sequential \SLOCALedge model seems to be extremely powerful. The only weak point of an algorithm is that it has to compute a valid solution for \emph{any} arrival order of the edges. While the computed solution can depend on that order (as long as it solves the desired problem), the algorithm cannot decide on the order. 

In order to obtain an efficient \LOCAL algorithm for a given \SLOCALedge algorithm $\mathcal{A}$ with locality~$\ell$, we compute a schedule to process the edges as follows. We compute a distance-$\ell$ edge coloring of $G$, i.e., edges that have the same color are at least $\ell$ hops apart. Such a coloring requires $O(\Delta^\ell)$ colors and can be computed in $O(\Delta^\ell+\log^*n)$ rounds. We define an order $\pi$ on the edges by ordering all edges according to their color, breaking ties arbitrarily for edges with the same color.  

Our final \LOCAL algorithm $\mathcal{B}$ will simulate the \SLOCALedge algorithm $\mathcal{A}$ according to the order $\pi$. Let us denote this latter algorithm by $\mathcal{A}_{\pi}$.
To simulate $\mathcal{A}_{\pi}$, algorithm $\mathcal{B}$  iterates through the color classes. Now, we compute the output of edges within one color class in parallel using $\mathcal A$ in $\ell$ rounds.  The reason this can be done in parallel is that any two edges with the same color need to have distance at least $\ell$ and hence their computed outputs cannot influence each other in $\mathcal{A}_{\pi}$. This process takes $O(\ell\cdot \Delta^\ell)$ rounds. 

\medskip

Applying \Cref{cor:complexity} to $(1+\eps)\Delta$-edge coloring would result in an $O(\Delta^5+\log^*n)$-round algorithm.  As described in \Cref{sec:tecOverviewRelationOnline}, we further optimize the online algorithm to allow more edges to be processed simultaneously. 
As a result we obtain the following theorem, where \Cref{thm:LOCALLBDelta} shows that it  is tight up to an $\tilde{O}(\Delta^2)$ factor by providing an $\tilde{\Omega}(\Delta^2+\log^* n)$ lower bound for the problem.

\begin{restatable}{theorem}{LOCALNoNDNoDS}
\label{thm:LOCALmainNoNDNoDS}
For any constant $\eps>0$, there is a deterministic \LOCAL algorithm to edge color any $n$-node graph with maximum degree $\Delta\geq c\sqrt{\log n}$ with $(1+\eps)\Delta$ colors in $O(\Delta^4+\log^* n)$ rounds, for some sufficiently large constant $c>0$.
\end{restatable}

\subsection{The \CONGEST Algorithm}
\label{sec:tecOverviewCongest}

\paragraph{High-level algorithm.}
By \Cref{thm:online_is_ESLOCAL}, we know that edges that are farther than $5$-hops apart can be processed simultaneously. To create a schedule that does this, we compute a \emph{distance-5 coloring}: edges that are within $5$ hops need to receive distinct colors. Such a coloring can be computed in $\poly(\Delta)$ \CONGEST rounds with $O(\Delta^5)$ colors. 

Now we execute this schedule: for any color, all edges of that color are processed according to the algorithm of \Cref{thm:online_is_ESLOCAL}. In the \LOCAL model, this is straightforward since any edge only need information from $5$ hops away. In the \CONGEST model, this immediately leads to congestion issues.

When an edge $e$ is processed it needs two things: its sampling probabilities $P_{ec}$ and a way to compute the change in potential affected by the edge $e$. The former is easy to maintain: whenever we process an edge, we inform the neighboring, unprocessed edges $f$ that they need to update their values $P_{fc}$. This leads to a congestion of $O(\Delta)$, since there are $O(\Delta)$ colors~$c$, which is small enough. 

The harder part is making the potentials for derandomization available. Hereto, the edge $e$ needs to know:\footnote{This is a simplification, for more details, see \Cref{lem:potential-gives-slocal}.}

\begin{itemize}
    \item All unprocessed edges within distance $5$;
    \item For each \emph{processed edge} $f$ within distance $5$: their sampling probabilities $P_{fc}$ (for all colors $c$) right before they are processed, i.e., colored. 
\end{itemize}
The former is relatively easy, and can be done in $O(\Delta^5)$ rounds at the start of the algorithm. The latter is more complicated. If edge $e$ requests this information at the time it needs it, this takes $O(\Delta^6)$ rounds. Multiplying this with the number of steps in the schedule gives a total running time of $O(\Delta^{11})$ -- much larger than we want. 

In the \LOCAL model, \Cref{thm:LOCALmainNoND} additionally uses degree splitting to first divide the input graph into many independent graphs with maximum degree $\Theta(\sqrt{\log n})$. Using degree splitting, the total running time becomes $\tilde O(\log^2 \Delta\log n)$ plus the time needed to color a graph with maximum degree $\Theta(\sqrt{\log n})$. 
\CONGEST degree splitting can not be found in the literature. In \Cref{sec:Congest_DS}, we show how to obtain it with minor adaptations of existing algorithms. 

Applying degree splitting to the algorithm as just described gives $O(\Delta^{11})=O(\log^{5.5} n)$ rounds.

\paragraph{Challenge 1: Reducing the number of rounds in the schedule.}
As already outlined in \Cref{sec:tecOverviewRelationOnline}, we do not need to schedule all edges within distance $5$ at different times, but only some conflicting edges, encoded in a \emph{conflict graph}. This conflict graph has degree $O(\Delta^4)$ -- a factor $\Delta$ less than the degree of $G^5$. This means we can color the edges of $G$ with $O(\Delta^4)$ colors, such that no two conflicting edges are the same color. Hence we can take this coloring as our schedule. To compute such a coloring, we are working with some virtual graph. This means we have some congestion issues in the \CONGEST model. In \Cref{sec:CONGEST_schedule}, we provide an efficient algorithm. 

\paragraph{Challenge 2: Reducing the number of messages.}
Our second challenge is that we need to send many messages to update the potential functions: an edge $e$ needs information from processed edges within 5 hops: up to $O(\Delta^5)$ many. Pulling this information when needed is relatively inefficient. 
We speed this up by \emph{pushing} the changes in the $5$-hop neighborhood, rather than pulling them. Suppose only one edge is updated at the time. When $f$ is colored, it can broadcast its $O(\Delta)$ values from the probability distribution to its 5-hop neighborhood in $O(\Delta)$ rounds. Then when $e$ is being processed, it already has the right information. The complication here is, that multiple edges are processed at the same time. In \Cref{sec:CONGEST_main}, we show that this leads to a congestion of $O(\Delta^3)$. This careful analysis uses again that we are scheduling with the conflict graph and not with $G^5$. 

By overcoming these two challenges, we achieve an $O(\Delta^7)$-round algorithm. 
To decrease this further, we need to consider the \emph{size} of the messages we are sending.

\paragraph{Challenge 3: Reducing the message size.}
As explained above, we have a schedule of $O(\Delta^4)$ steps, where in each step $O(\Delta^3)$ messages need to traverse each edge. Initially, those messages are of size $O(\log n)$. In \Cref{sec:CONGEST_main}, we will show that the size of a vertex identifier is the bottleneck here. 
Since edges only use information from their $5$-hop neighborhood, vertex identifiers do not need to be distinct beyond this distance. We can attain new identifiers satisfying this using a vertex coloring of the power graph from prior work. This reduces the size of the identifiers to $O(\log \Delta)$. Using that $\Delta = \Theta(\sqrt{\log n})$, we see that we can send $\Delta^2$ messages in $O(\log(\Delta))=O(\log\log n)$ rounds. This means that for each step of the schedule, we only need $O(\log^{0.5}n \log\log n)$ rounds. 
In total, this means that our algorithm takes $\tilde O(\log^{2.5}n + \log^2 \Delta \log n)$ rounds.



\newpage
\section{An Online Algorithm that is \SLOCALedge with locality 5}\label{sec:online_locality}

\newcommand\udeg{\mathsf{badness}}
\newcommand\baddeg{\mathsf{baddeg}}
\newcommand\Calg{\mathcal{C}}
\newcommand{\idx}[1]{^{(#1)}}
\newcommand{\idxz}{\idx{0}}
\newcommand{\idxt}{\idx{t}}
\newcommand{\idxtp}{\idx{t+1}}


In this section, we prove the following theorem.
\ThmOnlineisESLOCAL*

We provide the proof in two parts. As a warm-up, we first show that the randomized online algorithm of \cite{BlikstadSVW25} has locality $1$ when implemented in \SLOCALedge{}. Then we show how to derandomize that algorithm with \SLOCALedge{} locality $5$. Formally, the resulting algorithm is not an online algorithm as it accesses future edges.


\paragraph{\boldmath Implication for \SLOCAL.}
\Cref{thm:online_is_ESLOCAL} immediately implies the following result in \SLOCAL. 

\thmSLOCALmain*
\begin{proof}
    Any \SLOCALedge{}($\ell$) algorithm can be run on \SLOCAL with locality $\ell+2$. The reason is that the knowledge of edges in \SLOCAL can be stored at the endpoints. The distance between the endpoints of two edges is at most the distance between the edge plus two. Now \Cref{thm:online_is_ESLOCAL} gives a deterministic \SLOCALedge{}(5) algorithm for this edge coloring problem, so there is a deterministic \SLOCAL{}(7) algorithm. 
\end{proof}


\subsection{Warm-Up: Randomized with locality 1}
We begin by observing that the online algorithm of \cite{BlikstadSVW25} works in \SLOCALedge{} with locality $1$.
\begin{theorem}
\label{thm:slocal-rand}
    Let $\eps>0$ be a parameter. There is a randomized \SLOCALedge algorithm that, given a graph with maximum degree $\Delta =\Omega(\sqrt{\log n}/\eps^{16})$, colors it with $(1+\eps)\Delta$ colors with locality $1$. 
\end{theorem}

\paragraph{Online edge coloring of \cite{BlikstadSVW25}}
Below we present the pseudo-code of the randomized online edge coloring algorithm of \cite{BlikstadSVW25}, simplifying to the case that we have access to future edges. We skip the analysis\footnote{When we derandomize the algorithm in the next section, we will need to redo large parts of the analysis.} and refer to \cite{BlikstadSVW25} for a formal proof of correctness. 

\begin{lemma}[\cite{BlikstadSVW25}]
\label{lem:online-random-color}
    The online \cref{alg:edge-coloring} edge colors a graph, where edges arrive in adversarial order (from an oblivious adversary), with $\Delta + O(\eps \Delta)$ colors, with high probability in $n$.
\end{lemma}

For this section, we simply need to observe that their online algorithm almost straightforwardly also works in \SLOCALedge{}. This can be done without understanding the analysis and instead by just inspecting the pseudo-code. For convenience, we begin with a brief overview of the algorithm of \cite{BlikstadSVW25}.

\begin{algorithm}[ht!]
	\caption{Randomized Online Edge Coloring Algorithm \cite{BlikstadSVW25} (with access to future edges)}
	\label{alg:edge-coloring}
	\begin{algorithmic}[1]
      \Statex \underline{\smash{\textbf{Input:}}} Maximum degree $\Delta\in \mathbb{Z}_{\geq 0}$, and parameter $\eps \ge c_{\eps} \cdot (\frac{\sqrt{\log n}}{\Delta})^{1/16}$.
        \Statex 
                      \textbf{\underline{Initialization:}} 
        $\Calg \leftarrow [\Delta]$, 
        $A \gets \frac{c_{A}}{\eps^2\Delta}$,
        $\alpha \gets \eps^3/100$, and constants
        $c_\varepsilon := 10$, $c_A := 4$, $c_{K} := 35c_{A}^2$.
        \Statex\qquad\qquad\qquad\;\;\textbf{for each} $e \in E$ and $c \in \Calg$:\; Set $P^{(0)}_{ec} \leftarrow \frac{1 - \varepsilon}{\Delta}$.
        
        \Statex 
		\For{\textbf{each} online edge $e_t=\{u,v\}$ on arrival} 
            \State \textbf{for each} $f\in E$ and $c\in \Calg$: Set $P^{(t)}_{fc}\gets P^{(t-1)}_{fc}$. \Comment{May be overridden below if $f$ neighbors $e_t$}

            \If{$u$ or $v$ are \emph{bad} (i.e., have $\udeg \ge 2c_{K}\eps \Delta$)}
                \If{$P^{(t-1)}_{ce_t} = 0$ for all $c\in \Calg$, or either $u$ or $v$ have $\ge \alpha \Delta$ \emph{bad} neighbors}
                \State Mark $e_t$.
                    \label{line:obliv:mark_bad}
                \Else 
                    \State Assign to $e_t$ an arbitrary color $c \in \Calg$, for which $P^{(t-1)}_{e_tc} > 0$.
                    \label{line:obliv:color_bad}
                    \For{\textbf{each} $f \in E$ not arrived yet, such that $f \cap e_t \neq \emptyset$}
                        \State $P^{(t)}_{fc} \leftarrow 0$.
                        \label{line:obliv:burn}
                    \EndFor
                \EndIf
                \Else

            \If{$\sum_{c\in \Calg} P^{(t - 1)}_{e_tc} > 1$}
                Mark $e_t$.
                \label{line:obliv:mark_z_ge_1}
            \Else
                \State Sample $K_t$ from $\Calg \cup \{\perp\}$ with probabilities $\left(P^{(t-1)}_{e_t1},\dots,P^{(t-1)}_{e_t\Delta}, 1 - \sum_{c\in \Calg}P^{(t-1)}_{e_tc}\right)$.
                \label{line:obliv:sample}
                \If{$K_t\in \Calg$} Assign color $K_t$ to $e_t$. \label{line:obliv:color}
                \Else\, Mark $e_t$. \label{line:obliv:mark_bot}
                \EndIf
                \For{\textbf{each} $f \in E$ not arrived yet, such that $f \cap e_t \neq \emptyset$}
                    \State $P^{(t)}_{fK_t} \leftarrow 0$ \textbf{if} $K_t \in \Calg$
                        \Comment{Prevent $f$ from being colored $K_t$}
                    \label{line:obliv:zero}
                        \State $P^{(t)}_{fc} \leftarrow\frac{P^{(t-1)}_{fc}}{1 - P^{(t-1)}_{e_tc}}$
                        \textbf{for each}
                    $c\in \Calg\setminus\{ K_t\}$ \textbf{where} { $P^{(t-1)}_{fc} \leq A$}
                    \label{line:obliv:scaleup}
                    \label{line:obliv:star}
                \EndFor
                
            \EndIf
                \If{$e_t$ marked in Line~\ref{line:obliv:mark_z_ge_1} or Line~\ref{line:obliv:mark_bot}}
                    \State Increment $\udeg(u)$ and $\udeg(v)$. \label{line:obliv:increment_badness}
                    \EndIf

            \EndIf
            \If{$e_t$ marked}
            \State Color $e_t$ greedily using the separate palette $\{\Delta+1, \Delta+2, \ldots\}$.
                \label{line:obliv:greedy}
            \EndIf
		\EndFor
	\end{algorithmic}
\end{algorithm}

\paragraph{A brief overview of the algorithm.}
For each edge $e\in E$ and color $c\in [\Delta]$, it keeps track of a variable $P_{ec}$, intuitively representing the probability that $e$ should be colored with $c$, when the edge $e$ arrives. These are initially set to $\frac{1-\eps}{\Delta}$. The algorithm aims to color most edges using these $\Delta$ colors, leaving only $O(\eps \Delta)$ uncolored edges incident to each vertex, for which we fall back to greedy with a separate palette. Indeed, if the degree of marked edges is at most $\Delta' = O(\eps \Delta)$, greedy will use at most $2\Delta'-1 = O(\eps \Delta)$ additional colors.

When an edge $e$ arrives, the common case will be that we sample a color $K_t$ according to the probabilities $P_{ec}$, and assign this color to $e$. We must now update the probabilities $P_{fc}$ for neighboring edges~$f$. For example, we must set $P_{fK_{t}}\gets 0$, as we are not allowed to also color $f$ the same color $K_{t}$. Additionally, to preserve marginals, the algorithm naturally should scale up the probabilities for colors $c$ not chosen at $e$, specifically we scale up $P_{fc}$ by $\frac{1}{1-P_{ec}}$.

The algorithm explained so far is the same as the simplified version in the technical overview and would work as long as $\Delta = \Omega(\log n)$. \cite{BlikstadSVW25} introduces a few changes to make the algorithm well-defined, easier to analyze, and also work for lower maximum degree $\Delta = \Omega(\sqrt{\log n})$:

\begin{itemize}
\item In case $\sum_{c} P_{ec} > 1$, then these cannot be used as sampling probabilities. Instead the algorithm just leaves the edge uncolored and argues it does not happen so often. 
\item The algorithm thresholds the $P_{fc}$ by not scaling them up if they already exceed a carefully chosen threshold~$A$.
\item The algorithm counts, in the variables $\udeg(v)$, how many of $v$'s already processed edges are left uncolored. If this is too large (already $\Theta(\eps \Delta)$, it calls the vertex \emph{bad}, and handles its remaining incident edges in a different fashion by simply greedily assigning them available colors in Lines~\ref{line:obliv:mark_bad}--\ref{line:obliv:burn}.
\end{itemize}

Whenever $\Delta \ge \log n$, with high probability none of the above three things occur (that is, w.h.p., $\sum_{c} P_{ec} < 1$, and $P_{ec}\le A$, and no vertex becomes \emph{bad}). However, when $\Delta \approx \sqrt{\log n}$, each of these three events will occur, but quite rarely (intuitively, for a single edge or vertex, it occurs with probability $2^{-\Delta} \approx 2^{-\sqrt{\log n}}$). The thresholding of $P_{ec}$'s by $A$ makes sure that the error (or ``bad events'') does not spread in the graph. The \emph{bad} vertices will be reasonably spread out in the graph: in particular, each vertex will have at most $O(\eps \Delta)$ \emph{bad} vertices as neighbors. Intuitively, this means that the special-case handling of edges incident on bad vertices does not affect the rest of the algorithm too much.

\paragraph{\boldmath Implementing \cref{alg:edge-coloring} in \SLOCALedge{}.} Now we explain how to implement the algorithm in \SLOCALedge{} with locality $1$.  Instead of updating the values $P_{fc}$ in Lines~\ref{line:obliv:burn}, \ref{line:obliv:zero}, and \ref{line:obliv:scaleup}, we delay this until edge $f$ arrives (in which case we will have all the information to compute the values $P_{fc}$ in retrospect).
At each edge $e_t$ that has already arrived, we store:
\begin{itemize}
\item the time of arrival $t$;
\item the color assigned to it (either in $\Calg$ or in the greedy palette $\{\Delta+1, \Delta+2, \ldots\}$).
\item the values $P^{(t-1)}_{e_{t}c}$ for all colors $c$ (possibly used to sample the color);
\item the line of the algorithm it was assigned a color (either Line~\ref{line:obliv:color} or Line~\ref{line:obliv:color_bad}).
\end{itemize}
To process an arriving edge $e = \{u,v\}$, we first need to compute the values $\udeg(v), \udeg(u)$, and $P_{ec}$ for all colors $c$. We can do this by inspecting all the stored values on already arrived edges $f$ neighboring $e$. We then simply run the online algorithm to color $c$, and store the values as specified above on the edge. The locality of the algorithm is thus just $1$.

The correctness of the algorithm follows from \Cref{lem:online-random-color}. We note that our adversary in the \SLOCALedge model is an oblivious one, as the graph and arrival order is fixed in advance and does not depend upon the algorithm's randomness. This proves \cref{thm:slocal-rand}.

\subsection{Deterministic with locality 5}
\begin{theorem}\label{thm:slocal-det}
     Let $\eps>0$ be a parameter. There is a deterministic \SLOCALedge algorithm that, given a graph with maximum degree $\Delta =\Omega(\sqrt{\log n}/\eps^{16})$, colors it with $(1+\eps)\Delta$ colors with locality~$5$. 
\end{theorem}

\paragraph{Strategy.} The plan is to derandomize the randomized online/\SLOCALedge{} algorithm.\footnote{We note that \cite{BlikstadSVW25} already (implicitly) gives a \emph{deterministic} online algorithm, by (black-box) derandomizing their randomized algorithm \emph{against an adaptive adversary}. In this paper, we instead derandomize the randomized algorithm of \cite{BlikstadSVW25} that only works against a weaker \emph{oblivious adversary}, but also works for lower values of $\Delta$. We can do so in the \SLOCALedge{} model since the final graph is known to us (and not just the graph that arrived so far). We also care about the locality of the derandomization, unlike \cite{BlikstadSVW25}, which means that we need to be more careful with our approach.} This can be done by the use of \emph{pessimistic estimators}. This standard technique essentially replaces all the Chernoff bounds or uses of Azuma's inequality in the analysis of the randomized algorithms with potential functions. By combining these potential functions, one can obtain a single potential function $\Phi(t)$, parametrized by the time $t$, with the following properties:

\begin{itemize}
\item The initial value is $\Phi(0) < 1$.
\item In each step, there is a (deterministic) choice that will not increase the potential $\Phi$.
\item No ``bad'' events occur as long as $\Phi(t)<1$ (these are the same ``bad'' events that the Chernoff bounds or Azuma's inequality ruled out for the randomized algorithm---in our case this guarantees that the algorithm leaves few $O(\eps \Delta)$ edges uncolored incident to each vertex).
\end{itemize}

Now, a deterministic algorithm will just do the same as the randomized algorithm, except that whenever randomness is used it will instead simply pick some choice that does not increase the potential $\Phi$ (for example the choice that minimizes the potential). Concretely, in our algorithm, the only randomization is in the sampling step in Line~\ref{line:obliv:sample}. The deterministic algorithm will instead pick the color $K_{t}$ deterministically, by choosing the one that minimizes a potential $\Phi$ (which we will define later).

There are a few caveats with this approach. Firstly, computing the potential function $\Phi(t)$ might not be fast. In fact, for our edge coloring algorithm it would take $\approx n^{\Delta}$ time to compute the potential function in the sequential setting. However, in distributed models we are not restricted by computation time, as long as this computation only occurs inside the vertices, so this is not an issue for us.

Another potential problem is that the potential $\Phi$ is a global function, and so in our distributed setting one cannot fully compute it locally. Luckily, as we will see, whenever an edge $e$ arrives and the algorithm needs to deterministically choose a color for it (instead of randomly sample), it will be able to figure out the change in potential by only inspecting a local part of the graph of small distance away from the edge $e$. This property is crucial to our derandomization, and specific to the algorithm we analyse. On a high level, we have this property because all the Chernoff bounds and Azuma's inequalities in the analysis---in \cite{BlikstadSVW25}---of the original randomized algorithm were \emph{local}, in the sense that each individual application of Azuma/Chernoff only concerns random variables associated with edges and vertices close together in the graph.

We formalize this idea in the following lemma:

\begin{lemma}
\label{lem:potential-gives-slocal}
There is a potential function $\Phi$ such that:
\begin{itemize}
\item any algorithm that, upon the arrival of an edge, simulates \cref{alg:edge-coloring}, except instead of sampling colors $K_t$ (in Line~\ref{line:obliv:sample}) chooses one (deterministically) that does not increase the potential $\Phi$, will leave at most $O(\eps \Delta)$ edges uncolored incident to each vertex.
\item the potential can be decomposed as $\Phi(t) = \sum_{v\in V} \Phi_{v}(t)$, where each $\Phi_{v}$ can be computed from (local) knowledge of just the edges in the 3-hop neighborhood of vertex $v$. For such edge that arrives in the future, only their existence is needed. For each such edge $f$ that has already arrived, we need:
\begin{enumerate}[(1)]
    \item the time of arrival $t$;
    \item the color assigned to it;
    \item the values $P_{fc}^{(t)}$ (up to precision $\poly(\Delta)$); and 
    \item the line in the algorithm $f$ was assigned a color.
\end{enumerate}
\end{itemize}
\end{lemma}

The first item in the above lemma is not surprising as it just follows from standard pessimistic estimators. The second item is also crucial for our use-case, and follows from that not only the randomized algorithms decisions are local, but that the analysis, in \cite{BlikstadSVW25}, of the randomized algorithm is also local.
    With the second item, we mean that in order to compute $\Phi_{v}(t)$, one only need to know the following for (already arrived) edges $f$ that are at most distance $3$ away from $v$\footnote{We note that the $\Phi_{v}$ will depend on an edge $(u,w)$ if an edge $(v,u)$ exists in the graph, \emph{even if the edge $(v,u)$ has not arrived yet}.} of $v$: (1) the time of arrival $t$, (2) the color assigned to it, (3) the values $P_{fc}^{(t)}$, and (4) the line in the algorithm $f$ was assigned a color. 

\Cref{lem:potential-gives-slocal} immediately implies the deterministic \SLOCALedge{} algorithm of \cref{thm:slocal-det}:
\begin{proof}[Proof of \cref{thm:slocal-det}.]
    Suppose an edge $e = (u,v)$ arrives.
    First, from inspecting the values stored at all neighboring edges, the algorithm can compute $\udeg(u),\udeg(v)$, and  $P_{ec}$ for all colors $c$ (as in the randomized setting).  The deterministic algorithm then simulates \cref{alg:edge-coloring}. If the algorithm should sample a color $K_{t}$ in Line~\ref{line:obliv:sample}, we must replace this by a deterministic way of choosing the edge. By \cref{lem:potential-gives-slocal}, edge $e$ can only affect the potentials $\Phi_{w}$ for vertices $w$ at distance at most $2$ away from $e$. In order to compute the change in potential we thus need to collect all information on edges up to distance $5$ away (so the change in the $\Phi_{w}$'s, which are distance $2$ away and depend on their $3$-hop neighborhood, can be computed). See \cref{fig:twohop} for an illustration. All in all, since we only explore distance $5$ away from the edge $e$, the locality of our deterministic \SLOCALedge{} algorithm is $5$.
\end{proof}

\begin{figure}[!ht]
    \includegraphics[page=1,width=\textwidth]{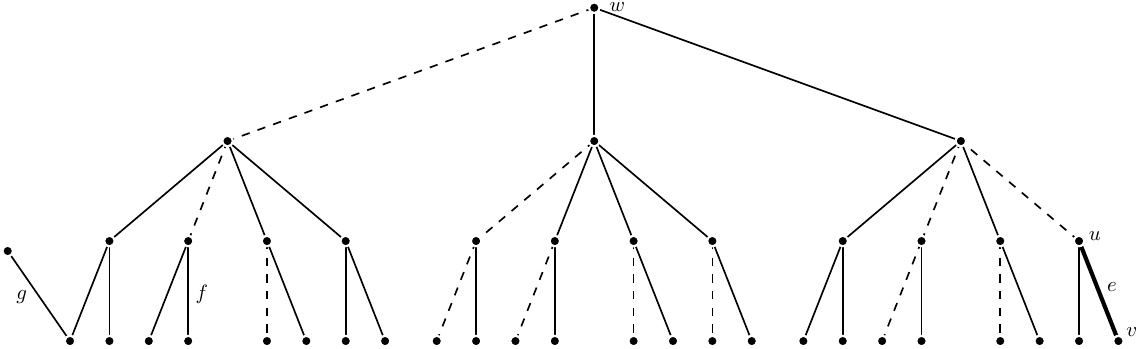}
    \centering
    \caption{At the time of arrival of edge $e = (u,v)$, some other edges have already arrived, while some (dashed) edges will arrive in the future. Edge $e$ can affect the potential $\Phi_{w}$ stored at $w$. In order to compute $\Phi_{w}(t)$, and see how it is affected by different choices of coloring $e$, all information on edges in the 3-hop neighborhood of $w$ (and thus $5$-hop neighborhood of $e$) needs to be collected. For example, information stored at edge $f$ is necessary to make the coloring choice at $e$, while information on edge $g$ is not necessary.}
    \label{fig:twohop}
\end{figure}

What remains, for this section, is to show how to construct a potential function as in \cref{lem:potential-gives-slocal}. Doing so is conceptually straightforward, but technically intricate, as we need to go over the $\approx30$-page analysis of \cref{alg:edge-coloring} in \cite{BlikstadSVW25} and replace each use of Chernoff bound or Azuma's inequality with an appropriate pessimistic estimator, and verify that the potentials we get can be computed locally.

\subsubsection{Invariants}

Before we begin setting up a potential $\Phi$, we define a few invariants which we will prove our final deterministic algorithm satisfies. Invariants~\labelcref{inv:fewbadcol}--\labelcref{inv:badvertexprop} each correspond to a key-lemma (of the same name) in the randomized analysis of \cite{BlikstadSVW25}.

\begin{definition}[Invariants]
    We analyze the following invariants at time $t$:
    \begin{enumerate}[(i)]
        \item\label{inv:potential} \textbf{Small Potential.} $\Phi(t) < 1$.
        \item\label{inv:fewbadcol} \textbf{Few Bad Colors.} For all edges $e$ (possibly not yet arrived), there are at most $2\eps^5 \Delta$ colors $c$ with $P_{ec}\idxt > A$.
        \item\label{inv:fewbadnei} \textbf{Few Bad Neighbors.} All vertices $v$ have at most $\alpha \Delta$ bad neighbors.
        \item\label{inv:badvertexprop} \textbf{Bad Vertex Property.}
            For each vertex $v$ that turns bad at time step $t_{0} < t$,
            there are at most $\eps \Delta$ neighbors $u$ of $v$, arriving between time steps $t_{0}$ and $t$, such that $Z\idx{t_{e_{t}-1}}_{e_{t}} = 0$, where
            $e_{t} = (u,v)$ and $t_{e_{t}} > t_{0}$ is the time step at which $e_{t}$ arrives.
    \end{enumerate}
\end{definition}

We will show that---assuming these invariants hold at the time an edge $e$ arrives---that there is a deterministic choice that makes these invariants also hold after $e$ is colored.
In fact, by setting up the potential properly, all invariants \labelcref{inv:fewbadcol}--\labelcref{inv:badvertexprop} will be implied by \labelcref{inv:potential}. Hence it will always suffice to make any coloring choice that does not increase the potential. The invariants are still useful for our analysis: that they hold at time $t$ will be crucial to show that there is a choice that avoids the potential $\Phi$ increasing at time step $t+1$. Moreover, the next lemma shows that the invariants guarantee that the algorithm uses few colors.

\begin{lemma}
    If the invariants are always satisfied, then the algorithm uses $\Delta + O(\eps \Delta)$ colors.
\end{lemma}
\begin{proof}
    It suffices to prove that all vertices $v$ have at most $\Delta' = O(\eps \Delta)$ incident edges \emph{marked} by the algorithm (and then colored greedily with $\le 2\Delta'$ colors).

    For each vertex $v$, it can only have at most $\eps \Delta$ incident edges marked either in Line~\ref{line:obliv:mark_z_ge_1} or Line~\ref{line:obliv:mark_bot}. This is since each time it is marked in those lines $\udeg(u)$ is incremented, and once $\udeg(v)\ge \eps \Delta$ (i.e., $v$ becomes \emph{bad}) incident edges $(v,u)$ never enters those lines.

    Thus it suffices to verify that Line~\ref{line:obliv:mark_bad} marks at most $O(\eps\Delta)$ edges incident to $v$. A vertex is \emph{dangerous} is it has $>\alpha \Delta$ bad neighbors, and invariant \labelcref{inv:fewbadnei} says this can never happen. While $v$ is good, by that same invariant, it has at most $\alpha \Delta$ bad neighbors, and thus at most $\alpha \Delta = O(\eps^3 \Delta) = O(\eps \Delta)$ incident edges to $v$ can be marked in Line~\ref{line:obliv:mark_bad} while $v$ remains good. If $v$ at some point becomes bad, we need to count the number of times the event $Z_{e_{t}c} = 0$ for arriving incident edges $e_{t}$. Invariant~\labelcref{inv:badvertexprop} guarantees this happens at most $\eps\Delta$ times.
\end{proof}

\subsubsection{Preliminaries: Martingales and Pessimistic Estimators}

\paragraph{Martingales.}
Suppose that $Z\idxz, Z\idx{1}, \ldots$ is a (super) martingale with step size $S$, that is:
\begin{itemize}
    \item \emph{Super-Martingale:} $E[Z\idxtp \mid Z\idxz,\ldots,Z\idxt] \le Z\idxt$
    \item \emph{Step Size:} There are $A_{t}$ and $B_{t}$ determined by $(Z\idxz, \ldots, Z\idxt)$ such that $A_{t} \le Z\idxtp-Z\idxt \le B_{t}$ and $B_{t}-A_{t} \le S$.
\end{itemize}

We then have the following classic concentration inequality:
\begin{theorem}[Azuma's Inequality]
    If $Z$ is a supermartingale with step size $S$, then for any $N$:
    \begin{equation*}
        \Pr[Z\idx{N}-Z\idxz\ge \lambda] \le \exp\left(-\frac{2\lambda^2}{S^2N}\right).
    \end{equation*}
\end{theorem}
\begin{remark}[Chernoff Bound]
    If we have variables $X\idxt$ that are $\{0,1\}$ and independent,
    we can set $Z_{\idxt} = \sum_{i=1}^t (X_{i} - E[X_{i}])$.
    We see that
    $Z\idx{N}-Z\idxz = \sum_{i=1}^N X\idxt$, and that the step size $S = 1$, so Azuma's inequality recovers a standard additive Chernoff-Hoeffding bound
    $\Pr[\sum_{i=1}^N X\idx{i}\ge \E[\sum_{i=1}^{N}X\idx{i}] +\lambda] \le \exp\left(-2\lambda^2/N\right)$.
\end{remark}

\paragraph{Derandomizing Azuma's inequality.}
We now set up a potential function in order to derandomize Azuma's inequality (and Chernoff-Hoeffding bounds) using pessimistic estimators.

\begin{definition}
For parameters $t\ge 0$ and $\lambda, S, N > 0$, we define for any $X$:
\begin{align*}
\phi(X, t, \lambda, S, N)
    = \exp\left(\frac{4\lambda}{S^2N}\left(X - \frac{\lambda}{2}\left(1+ \frac{t}{N}\right)\right)\right).
\end{align*}
\end{definition}

The intuition behind the function $\phi$ is the following: Set
\begin{align*}
    \Phi_Z(t) = \phi(Z\idxt - Z\idxz, t, \lambda, S, N).
\end{align*}
This makes for a good potential function controlling the concentration of $Z\idxt$ for $t\le N$.  Indeed, if $\Phi_Z\idxt$ is small, then $Z\idxt - Z\idxz$ must also be small. More precisely we have the following lemma.

\begin{lemma}
    \label{lem:potential-implies-concentration}
If $\phi(X, t, \lambda, S, N) < 1$ and $t\le N$, then
$X < \lambda$.
\end{lemma}
\begin{proof}
    If $t\le N$, then $\phi(X, t, \lambda, S,N) \ge \exp(\frac{4\lambda}{S^2 N}(X-\lambda))$ by definition. So, if $X \ge \lambda$, then
    $\phi(X, t, \lambda, S,N) \ge \exp(0) = 1$, contradicting the choice of $X$.
\end{proof}
As we will now see, when the next step of the martingale $Z\idxtp$ is randomly sampled (given the conditioning on $Z\idxz, \ldots, Z\idxt$), then the potential $\Phi_{Z}\idxtp$ does not increase in expectation. Here we make use of the assumption on the step size.





\begin{lemma}
If $X$ is a random variable such that $\E[X] \le x$ and 
$a \le X \le b$ where $b-a\le S$, then
\begin{align*}
\E[\phi(X, t+1, \lambda, S, N)]
\le \phi(x, t, \lambda, S, N)
\end{align*}
\end{lemma}
\begin{proof}
By Hoeffding's lemma, for any $y\ge 0$, we have
\begin{align*}
\E[\exp(y(X - x))] \le \exp(y^{2}S^2/8)
\end{align*}
Set $y = \frac{4\lambda}{S^2N}$, and we get 
\begin{align*}
\E\left[\exp\left(\frac{4\lambda}{S^2N}(X - x)\right)\right] \le \exp\left(\frac{2\lambda^2}{S^2N^2}\right).
\end{align*}
We finish the proof by rearranging the terms and applying the above inequality:
\begin{align*}
\E\left[
\frac{\phi(X, t+1, \lambda, S, N)}{\phi(x, t, \lambda, S, N)}
\right]
&=
\E\left[
    \exp\left(\frac{4\lambda}{S^2N}\left(X - \frac{\lambda}{2}\left(1+ \frac{t+1}{N}\right)\right)\right)
     / 
    \exp\left(\frac{4\lambda}{S^2N}\left(x - \frac{\lambda}{2}\left(1+ \frac{t}{N}\right)\right)\right)
\right]
\\ &=
\E\left[
    \exp\left(\frac{4\lambda}{S^2N}\left(X - x\right) - \frac{2\lambda^2}{S^2N^2}\right)
\right]
 \le 1.\qedhere{}
\end{align*}
\end{proof}

The above lemma thus implies that our martingale-controlling potential $\Phi_{Z}(t)$ is non-increasing in expectation. This means that a deterministic algorithm can always instead choose the random bits in such a way that the potential $\Phi_{Z}$ is non-increasing. If it started out smaller than one, then it will end smaller than one, and thus guaranteeing concentration of $Z\idxt$ for all $t\le N$. We summarize our findings for $\Phi_{Z}$ in the following theorem, which immediately follows from the two preceding lemmas.

\begin{theorem}
    \label{lem:potential-of-martingale}
    Suppose $Z\idxz, Z\idx{1},\ldots$ is a supermartingale with step size $S$, and that we are given a parameter $\lambda > 0$. We define a potential:
    \begin{equation*}
        \Phi_Z(t) = \phi(Z\idxt - Z\idxz, s_{Z}(t), \lambda, S, N),
    \end{equation*}
    where $s_{Z}(t)$ counts the number of non-trivial steps the martingale has taken up to time $t$. That is $s_{Z}(0) = 0$, and whenever $Z\idxtp \neq Z\idxt$, then $s_{Z}(t+1) = s_{Z}(t)+1$ (for trivial steps where $Z\idxtp=Z\idxt$, we simply let $s_{Z}(t+1)=s_{Z}(t)$). We set $N$ above as $s_{Z}(\infty)$, i.e., the total number of non-trivial steps.
Then, this potential satisfies the following properties:
\begin{itemize}
    \item If $\Phi_{Z}(t) < 1$, then $Z\idxt -Z\idxz < \lambda$.
    \item Initially $\Phi_{Z}(0) = \phi(0, 0, \lambda, S, N) = \exp(-\frac{2\lambda^2}{S^2 N})$. 
This term exactly corresponds to the bound in Azuma's Inequality.
    \item $\Phi_{Z}(t)$ is, in expectation, non-increasing. More formally: suppose $Z\idxz,\ldots, Z\idxt$ are measurable (i.e., all random bits have---deterministically---already been resolved for them). Then $E[\Phi_{Z}(t+1)] \le \Phi_{Z}(t)$, where the expectation is over the remaining random bits for $Z\idxtp$.
\end{itemize}
\end{theorem}

\paragraph{Derandomizing union bounds.}
The preceeding lemma about the non-increasing expectation of the potential is very useful in combination with \emph{linearity of expectation}. Say we have (polynomially) many martingales
$(Z_1)\idx{t\ge0}, (Z_2)\idx{t\ge0}, \ldots, (Z_{\poly(n)})\idx{t\ge0}$.
Say, for each $i$, Azuma's inequality gives that the $i$'th martingale $Z_{i}$ is concentrated with high probability (in $n$) if we run some randomized algorithm. A union bound then tells us that with high probability \emph{all} such martingales are simultaneously concentrated.

We can derandomize the union bound by creating a global potential function by simply summing up the individual potential functions for the martingales:
\begin{align*}
\Phi(t) = \sum_{i} \Phi_{Z_{i}}(t).
\end{align*}
Now, by linearity of expectation, 
$\Phi(t)$ is also non-increasing (in expectation). Again, whenever the randomized algorithm needs to use random bits, a deterministic algorithm can choose them in such a way that $\Phi(t)$ is non-increasing. Thus, if $\Phi(0) < 1$, then $\Phi(t) < 1$ too, 
and in particular $\Phi_{Z_{i}}(t) < 1$ for each martingale $Z_{i}$ (since the potentials are all non-negative). By \cref{lem:potential-implies-concentration}, we thus have deterministically guaranteed the concentration bounds $Z\idxt_{i} -Z\idxz_{0} \le \lambda$ for all $i$. What remains is to verify that $\Phi(0) < 1$. This is just 
$\sum_{i}\Phi_{Z_{i}}(0)$, and as we explained previously these terms correspond to exactly the (probability) guarantees Azuma's inequality gives, and hence if Azuma + Union bound guaranteed probability $<1$ of failure for the randomized algorithm, then
$\Phi(0) < 1$ for the deterministic algorithm too.

\subsubsection{\texorpdfstring{Definitions from \cite{BlikstadSVW25}}{Definitions from [BSVW25]}}
Before setting up the potential, we will need the following definitions from \cite{BlikstadSVW25}. The following definitions are taken (almost) verbatim from \cite{BlikstadSVW25}.

\begin{definition}[Parameters] \label{def:parameters-obliv}
    Define $\varepsilon := c_\varepsilon \cdot \left( \frac{\sqrt{\ln n}}{\Delta} \right)^{1/16}$ and $A := \frac{c_A}{\varepsilon^2 \Delta}$. Here, $c_\varepsilon := 10$ and $c_A := 4$ are the same constants as before. Additionally, let $c_K := 35c^2_A$ be a new constant,
    and let $\alpha := \varepsilon^3 / 100$. 
\end{definition}

\begin{definition}[Bad and Dangerous Vertices] \label{def:bad_vertices_etc}
    Consider a fixed instance of the online edge coloring problem, and let $v \in V$ be an arbitrary vertex. Let $\udeg^{(t)}(v)$ be the value of $\udeg(v)$ at time step $t$, during the execution of \cref{alg:edge-coloring}. Vertex $v$ is \emph{bad} (at time step $t$), if $\udeg^{(t)}(v) \geq 2c_K \cdot  \varepsilon \Delta$. Else, $v$ is \emph{good}.

    Further, for any vertex $v \in V$ and time step $t$, let $\baddeg^{(t)}(v)$ be the number of neighbors $w$ of $v$ such that: the edge $f := \{v,w\}$ arrived at time $t_f$ and vertex $w$ was \emph{bad} at time step $t_f$. This is the number of neighbors of $v$, taken up to time step $t$, which were \emph{bad} when they were connected to $v$. 
    Vertex $v$ is \emph{dangerous} (at time $t$) if $\baddeg^{(t)}(v) \geq \alpha \Delta$.
\end{definition}

In the following, fix an instance of the online edge coloring problem, and an edge $e = \{u,v\}$ with arrival time $t_e$. Consider the edges $\delta(e)$ that intersect with $e$, and which arrive before $e$.
\begin{definition} \label{def:classifying_edges}
    We partition the edges of $\delta(e)$ into three subsets $\delta(e)_{\text{good}}$, $\delta(e)_{\text{bad}}$ and $\delta(e)_{\text{rest}}$, as follows:
    \begin{itemize}
        \item It holds that $f \in \delta(e)_{\text{good}}$ if, at the time of arrival, none of the two vertices to which $f$ is incident are \emph{bad}.
        \item It holds that $f \in \delta(e)_{\text{bad}}$ if $f$ is incident to a bad vertex on arrival, but this bad vertex is not $u$ or $v$.
        \item It holds that $f \in \delta(e)_{\text{rest}}$ if $f$ is incident to a bad vertex on arrival, and this vertex is $u$ or $v$.
    \end{itemize}
\end{definition}
\begin{definition} \label{def:Z_extension}
    Define $Z_{e} = \sum_{c\in [\Delta]} P\idxt_{ec}$. Also,
    for a subset $C \subseteq \Calg$, define $Z^{(t)}_{eC} = \sum_{c \in C} P^{(t)}_{ec}$.
\end{definition}
\begin{definition} \label{def:lb_on_Z_extended}
For each color $c\in \Calg$, time step $t$ and non-arrived edge $f\in E$ at this time step, let
\begin{equation*}
\overline{P^{(t)}_{fc}}
=
\begin{cases}
P^{(t-1)}_{fc} / (1-P^{(t-1)}_{e_tc}) & \text{if in Line~\ref{line:obliv:scaleup} no assignment to $P^{(t)}_{fc}$ is made because of ${P^{(t-1)}_{fc} > A}$}\\
P^{(t)}_{fc} & \text{otherwise.}
\end{cases}
\end{equation*}
    Define also, for any fixed subset $C \subseteq \Calg$:
    \begin{equation*}
        \overline{Z^{(t)}_{eC}} = \sum_{c \in C} \overline{P^{(t)}_{ec}} \text{ \ \ and \ \ } Y^{(t)}_{eC} 
        = \sum_{t'=1}^t 1[e_{t'} \in \delta(e)_{\text{good}}] \cdot \left( \overline{Z^{(t')}_{eC}} - Z^{(t'-1)}_{eC} \right).
    \end{equation*}
\end{definition}

\begin{definition} \label{def:introducing_Q}
    For $w\in V$, let $U_{w} = \{(f,t_{f}) : \text{edges $e$ incident to $w$}\}$, where
    $t_{f}$ is the arrival time of edge $f$ (possibly in the past or future).
    For $w\in V$, any $(f,t_f) \in U_w$, and for the fixed color $c \in \Calg$, let:
    \begin{equation*}
        R^{(t)}_{fc} = P^{(t)}_{fc} \cdot \prod_{\substack{g \in U_w \\ t_g \leq t}} \left( 1 - P^{(t_g-1)}_{gc} \right).
    \end{equation*}
    Finally, define $Q^{(t)}_{U_{w}c} = \sum_{f \in U_w} R^{(\min\{t,t_f - 1)\}}_{fc}$.
\end{definition}

\paragraph{Complimentary definitions.} Here we also define some new properties.
Fix a vertex $w$, and a subset $U\subseteq N(w)$. Let $G[U]$ be the graph containing
all edges incident to a vertex in $U$. The \SLOCALedge{} algorithm can by exploring the $3$-hop neighborhood from $w$, find all edges in $G[U]$. Since $G$, and thus $G[U]$, has max-degree $\Delta$, the edges of $G[U]$ can be covered by $\Delta+1$ matchings. Define 
$\mathcal{M}_{w,U} = \{M_{1}, \ldots, M_{\Delta+1}\}$ to be a canonical way of partitioning the edges $G[U]$ into matchings.

\begin{definition}
    For a matching $M$ and a subset of colors $C$, define
    $K_{MC} = \sum_{e\in M}\sum_{c\in C} P_{ec}$.
\end{definition}

\subsubsection{Setting up a Complicated Potential}

We now set up our potential as a sum of potentials guaranteeing each of the invariants.
For a vertex $w$, we will define
\begin{equation}
    \Phi_{w} =
    \Phi_{w, \text{few bad colors}} + 
    \Phi_{w, \text{few bad neighbors}} +
    \Phi_{w, \text{bad vertex property}},
\end{equation}
and our global potential function as:
\begin{equation}
    \Phi = \sum_{w\in V} \Phi_{w}.
\end{equation}

In particular, we will see that $\Phi(0) < 1$, and that $\Phi_{w}$ can be computed locally by just inspecting the 3-hop neighborhood of $w$. The $\Phi_{w, \ldots}$ potentials will be defined in such a way that if $\Phi_{w,\ldots}(t) < 1$, then the corresponding invariant must hold at time $t$. Moreover, $\Phi_{w}$ will just be a sum of martingale-inspired potential of \cref{lem:potential-of-martingale}, and hence in expectation is non-increasing. There must thus always exist a deterministic choice that makes $\Phi(t+1)\le \Phi(t)$, whenever edge $e_{t}$ should be colored, proving \cref{lem:potential-gives-slocal}.

In particular, we prove the following three similar lemmas:

\begin{lemma}[Few Bad Colors]
    \label{lem:fewbadcol}
The potential $\Phi_{w, \text{few bad colors}}$ satisfies:
\begin{itemize}
    \item as long as $\Phi(t) < 1$, invariant~\labelcref{inv:fewbadcol} must hold.
        \item  $\sum_{w}\Phi_{w, \text{few bad colors}}(0) \le 2^{-50\Delta}n^{-50}$.
        \item  $\Phi_{w, \text{few bad colors}}$ is non-increasing in expectation.
        \item  $\Phi_{w, \text{few bad colors}}(t)$ can be computed by inspecting the $2$-hop neighborhood of $w$.
\end{itemize}
\end{lemma}
\begin{lemma}[Few Bad Neighbors]
    \label{lem:fewbadnei}
The potential $\Phi_{w, \text{few bad neighbors}}$ satisfies:
\begin{itemize}
    \item as long as $\Phi(t) < 1$ and invariant~\labelcref{inv:fewbadcol} holds, invariant~\labelcref{inv:fewbadnei} must also hold.
        \item  $\sum_{w}\Phi_{w, \text{few bad neighbors}}(0) \le 2^{-50\Delta}n^{-50}$.
        \item  $\Phi_{w, \text{few bad neighbors}}$ is non-increasing in expectation.
        \item  $\Phi_{w, \text{few bad neighbors}}(t)$ can be computed by inspecting the $3$-hop neighborhood of $w$.
\end{itemize}
\end{lemma}
\begin{lemma}[Bad Vertex Property]
    \label{lem:badvertexprop}
The potential $\Phi_{w, \text{bad vertex prop}}$ satisfies:
\begin{itemize}
    \item as long as $\Phi(t) < 1$ and invariant~\labelcref{inv:fewbadnei} holds, invariant~\labelcref{inv:badvertexprop} must also hold.
        \item  $\sum_{w}\Phi_{w, \text{bad vertex prop}}(0) \le 2^{-50\Delta}n^{-50}$.
        \item  $\Phi_{w, \text{bad vertex prop}}(0)$ is non-increasing in expectation.
        \item  $\Phi_{w, \text{bad vertex prop}}(t)$ can be computed by inspecting the $3$-hop neighborhood of $w$.
\end{itemize}
\end{lemma}

If all these three lemmas hold, then \cref{lem:potential-gives-slocal} follows, and thus also \cref{thm:slocal-det}.

\paragraph{Bit-complexity.} For our \CONGEST algorithms it is important that the messages are not too large. In particular, if we naively stored the $P_{fc}$'s exactly, they might require up to $\approx \Delta^\Delta \gg 2^{\sqrt{\log n}}$ bits to represent. This is way too much to send on an edge in the \CONGEST model. We note below that it is sufficient to store them to $O(\log \Delta)$ bits of precision.
 This is extra useful since if $\Delta \approx \log n$ a \CONGEST algorithm can send $O(\log n / \log \log n) = \tilde{O}(\Delta^2)$ many of these variables over a single edge in a single round.

\begin{lemma}[Rounding the $P_{fc}$'s]
    The $P_{fc}$'s can be rounded to the nearest multiple of $1/\Delta^{10}$ after they are computed, and the algorithm will still work.
\end{lemma}
\begin{proof}[Proof Sketch.]
    In the randomized algorithm, say the values $P_{ec}$, after they are computed when edge $e$ is processed, are randomly rounded up or down to a multiple of $1/\Delta^{10}$ (preserving the expectation). For the deterministic algorithm, this randomized rounding up or down can be derandomized similarly to the rest of the algorithm using the potentials. Since the rounding up or down preserves the expectation if done randomly, all our mortingales that are affected by this rounding can now be seen as additionally taking a small (of the order $1/\Delta^5$) martingale-step. The martingale-controlling potentials could thus simply see this as some additional tiny steps, that are dominated by the (larger) steps in the analysis, so they do not change much. The deterministic algorithm then chooses to either round up or down to minimize the new potential. As each martingale-controlling potential (see the analysis) is affected by at most $\Delta^3$ of the $P_{fc}$'s these extra steps barely change their initial values.
\end{proof}

\subsubsection{Few Bad Colors.}

For a (super-)martingale $Z$ with step size $S$ and number of non-trivial steps $N$, we define the potential $\Phi_{Z}$ as in \cref{lem:potential-of-martingale}.
For a (sub-)martingale $Z$, we note that $-Z$ is a super martingale, and can define $\Phi_{-Z}$ analogous.

The goal is to show
\begin{equation}
    Q_{U_{w}C} \ge \eps^5 \Delta - \frac{\eps^{6}\Delta}{2},
    \label{eq:QUC}
\end{equation}
for all vertices $w$ with edge-neighborhood $U_{w}$, and subsets of colors $C$ of size $|C| = \eps^5 \Delta$. Indeed, if we can show the above it implies the ``Few Bad Colors'' invariant:
From \cite[Lemma~3.19]{BlikstadSVW25}, we know that a color $c$ is not bad if $Q_{U_{w}\{c\}} \ge 1 - \eps/2$. If \labelcref{inv:fewbadcol} was false, it would thus imply the existence of an edge $e$ for which a set of $\eps^5\Delta$ colors where bad, but choosing those to be $C$ we get a contradiction with \cref{eq:QUC}. Note that initially,
$Q_{U_wC}(0) = \eps^5\Delta -\eps^6 \Delta$, so we can afford a deviation of $\eps^6\Delta/2$.

For the randomized algorithm, \cite{BlikstadSVW25} argues that \cref{eq:QUC} holds with high probability by the use of Azuma's inequality and union bounds.
\begin{claim}[{\cite[Lemma 3.24]{BlikstadSVW25}}]
    $Q_{U_{w}C}\idxt$ is a super-martingale with step size $24A$, and at most $\Delta^2$ non-trivial steps.
\end{claim}

We derandomize this with a pessimistic estimator, using parameter $\lambda = \eps^6\Delta/2$ and \cref{lem:potential-of-martingale}:
\begin{align*}
    \Phi_{w, \text{few bad colors}} = \sum_{C\subseteq[\Delta], |C| = \eps^5 \Delta} \Phi_{Q_{U_{w}C}}
\end{align*}
\cref{lem:potential-of-martingale} thus gives us, via a simple calculation and our choice of parameters:
\begin{align*}
    \Phi_{Q_{U_wC}}(0) = \exp\left(-\frac{2\lambda^2}{S^2 N}\right)
    = \exp\left(-\frac{2(\eps^6 \Delta)^2}{(24A)^2 \Delta^2}\right)
    = \exp\left(-\frac{2\eps^16 \Delta^2}{c_{A}^2 24^2}\right)
    \le \frac{1}{2^{100\Delta}n^{100}}
\end{align*}
This exactly corresponds to the probability calculation at the end of \cite[Proof of Lemma~4.10]{BlikstadSVW25} (see that paper for more details).

The contribution to $\Phi(0)$ of this term is thus (since there are $2^\Delta$ subsets of $[\Delta]$):
\begin{equation}
    \sum_{w}\Phi_{w, \text{few bad colors}} = \sum_{w}\sum_{C\subseteq[\Delta], |C| = \eps^5 \Delta} \Phi_{Q_{U_{w}C}} \le 2^{-50 \Delta} n^{-50}
\end{equation}

To conclude, the discussion above has shown \cref{lem:fewbadcol}; that is:
\begin{itemize}
    \item
 as long as $\Phi(t) < 1$, we know that the invariant~\labelcref{inv:fewbadcol} must hold.
        \item  $\sum_{w}\Phi_{w, \text{few bad colors}}(0) \le 2^{-50\Delta}n^{-50}$.
        \item  $\Phi_{w, \text{few bad colors}}$ is non-increasing in expectation.
        \item  $\Phi_{w, \text{few bad colors}}(t)$ can be computed by inspecting the $2$-hop neighborhood of $w$. Indeed, this follows from inspecting the definition of $Q_{U_{w}C}$, whichonly is updated when an edge in $w$'s $2$-hop neighborhood arrives.
\end{itemize}

We now proceed to prove analogous properties for the two other invariants.

\subsubsection{Few Bad Neighbors.}
We define:
\begin{align*}
    \Phi_{w, \text{few bad neighbors}} &= \sum_{U\subseteq N(w), |U| = \alpha \Delta} \Phi_{w,U}\\
    \Phi_{w, U} &= \sum_{M_{i}\subseteq \mathcal{M}(w,U)} \Phi_{w,U,M_{i}} + \Phi_{H_{U}} \\
    \Phi_{w, U,M_{i}} &= \sum_{M\subseteq M_{i}, |M| = |\eps\alpha\Delta|} \Phi_{M,[\Delta]}.
\end{align*}
(we will define $\Phi_{H_{U}}$ and $\Phi_{M,[\Delta]}$ below soon).

We will set things up so that $\Phi_{w,U} < 1$ will imply that there must be some vertex in $U$ that is not bad. Thus, if $\Phi(t) < 1$ then we know that $w$ has at most $\alpha \Delta$ bad neighbors, which is what we want to prove for invariant~\labelcref{inv:fewbadnei}.

For this, the proof in \cite{BlikstadSVW25} proceeds as follows. If we can show
that $\sum_{u\in U} \baddeg(u) < 2c_{K}\eps\alpha\Delta^2$, then we are done, as not all vertices in $u$ could be bad (i.e., have $\baddeg(u) \ge 2c_{K}\eps\Delta$).
We now look at all edges incident to some vertex in $U$, and count how many of them contributed to badness. Say the edge $e$ arrives. If any endpoint of $e$ is bad, it cannot contribute to the badness, so suppose both endpoints are good. We then split up into two cases:
\begin{itemize}
    \item $Z_{e}\in [1-c_{K}\eps, 1]$
    \item $Z_{e}\not\in [1-c_{K}\eps, 1]$
\end{itemize}
We argue that the second case rarely happens, and that the first case gives rise to a well-concentrated Chernoff-style term. In particular, define
$H_{U}\idxz = 0$, and whenever an edge $e_{t}$ in the former case arrives ($Z_{e_t}\in [1-c_{K}\eps, 1]$, and no endpoint of $e$ is \emph{bad}),
we update $H_{U}\idxtp = H_{U}\idxt - c_{K}\eps + 1[e_{t}\text{ was marked}]$.
By definition $E[H_{U}\idxtp] \le H_{U}\idxt$, since the edge is marked with probability $1-Z_{e} \le c_{K}\eps$. Thus $H_{U}$ is a super-martingale\footnote{What follows is really just a Chernoff bound with coupling, there are no complicated correlations here.} with at most $\Delta|U| = \alpha\Delta^2$ (non-trivial) steps, and it has step size $1$. Using $\lambda = \eps\alpha\Delta^2$:
we can define the potential $\Phi_{H_{U}}$ using \cref{lem:potential-of-martingale} with the following properties:
\begin{itemize}
    \item Initial value $\Phi_{H_{U}}(0) = \exp(-\frac{2(\eps\alpha\Delta^2)^2}{\alpha\Delta^2}) = \exp(-2\eps^2\alpha\Delta^2) \le 2^{-100\Delta}n^{-100}$.
    \item If $\Phi_{H_{U}} < 1$: then $H_{U} \le \lambda$, or equivalently
        at most $c_{K}\alpha\eps \Delta^2+\lambda = (c_{K}+1)\alpha\eps\Delta^2$ many edges of the first type were marked and increased badness.
\end{itemize}

For edges of the second type, $Z_{e}\not\in [1-c_{K}\eps, 1]$, we instead argue that there are at most $2\eps\alpha \Delta$ such edges in each matching $M_{i}\in \mathcal{M}(w,U)$ (assuming $\Phi(t) < 1$). We prove this in more generality to reuse it later:

\begin{lemma}
    \label{lem:matchings}
    Suppose $K_{MC} = \sum_{e\in M, c\in C} P_{ec}$ where $M$ is a matching and $C$ a subset of colors satisfying $|M|\cdot |C|^2 \ge \eps^{10} \Delta^3$.
    Then there is a potential function $\Phi_{M,C}$ such that:
    \begin{itemize}
        \item Initially $\Phi_{M,C}(0) < 2^{-90\Delta}n^{-90}$.
        \item In expectation, $\Phi_{M,C}$ is non-increasing.
        \item $\Phi_{M,C}(t)$ can be computed by inspecting the variables stored only at edges $f$ incident to some edge $e\in M$.
        \item If $\Phi_{M,C}(t) < 1$ then:
            \begin{itemize}
                \item $K_{MC}\idxt \le \frac{|M|\cdot|C|}{\Delta}(1-\frac{\eps}{2})$.
                \item Either (1) $K_{MC}\idxt \ge \frac{|M|\cdot|C|}{\Delta} - c_{K}\eps|M|$, (2) some endpoint of $M$ is \emph{bad}, or (3) invariant~\labelcref{inv:fewbadcol} does not hold.
            \end{itemize}
    \end{itemize}
\end{lemma}

Before proving \cref{lem:matchings}, we show how \cref{lem:fewbadnei} follows from it.

$\Phi_{H_{U}}$ only depend on the $2$-hop neighborhood of $w$, while $\Phi_{M,[\Delta]}$ depend on the $3$-hop neighborhood (for all the $M$s and $U$s we consider).
Moreover,  the total contribution of $\Phi_{w, \text{few bad neighbors}}$ can be bounded by $2^{-50\Delta}n^{-50}$, as there are at most $2^{\Delta}$ submatchings of $M_{i}$, and at most $2^{\Delta}$ subsets $U$ of $N(w)$. What remains is to show that \labelcref{inv:fewbadnei} holds as long as both $\Phi(t) < 1$ and invariant~\labelcref{inv:fewbadcol} holds.

Suppose for the sake of a contradiction that $U$ is a set of $\alpha \Delta$ bad neighbors. Since $\Phi(t) < 1$ we have both $\Phi_{H_{U}} < 1$ and $\Phi_{w,U,M_{i}} < 1$ for all $M_{i}$ in the partition of $G[U]$ into matchings $\mathcal{M}(w,U)$. The former implies that there are at most $(c_{K}+1)\alpha\eps\Delta^2$ many edges $e$ where $Z_{e}\in [1-c_{K}\eps, 1]$ that could have been left unmarked and incremented $\sum_{u\in U}\udeg(u)$. On the other hand, $\Phi_{w,U,M_{i}} < 1$ will imply that there are at most $2\eps\alpha\Delta$ edges in $M_{i}$, with both good endpoints (else it cannot increment $\udeg$), where $Z_{e}\not\in[1-c_{K}\eps,1]$.

Indeed, for each submatching $M\subseteq M_{i}$ with all good endpoints, of size $|M| = \eps\alpha \Delta$, and for $C = [\Delta]$, from the above lemma we know that
\begin{align*}
    \sum_{e\in M} Z_{e} &= K_{MC} \ge \frac{|M|\cdot|C|}{\Delta}-c_{K}\eps |M| = |M|(1-c_{K}\eps).\\
    \sum_{e\in M} Z_{e} &= K_{MC} \le \frac{|M|\cdot|C|}{\Delta} - \frac{\eps}{2} |M| = |M|(1-\frac{\eps}{2}).
\end{align*}
This proves that not all edges in each such submatching $M$ can have $Z_{e} < 1-c_{K}\eps$, and not all edges can have $Z_{e}> 1$. Indeed, this means that there are at most $2\eps\alpha\Delta$ edges (with good endpoints), per matching $M$ that can violate $Z_{e}\in [1-c_{K}\eps, 1]$.

Each of those can be incident to at most two vertices in $U$, and hence increase $\sum_{u\in U}\udeg(u)$ by at most $2$ each, for a total of $4\eps\alpha\Delta(\Delta+1)$ for all the matchings $M_{i}$. We thus bound:
\begin{align*}
\sum_{u\in U}\udeg(u)
\le
    (1+c_{K})\eps\alpha\Delta^2 + 4\eps\alpha\Delta(\Delta+1) < 2c_{K}\eps\alpha\Delta^2,
\end{align*}
contradicting that all vertices in $U$ where bad.
This proves \cref{lem:fewbadnei}.

\paragraph{Arguing about matchings.} We now show that \cref{lem:matchings} holds. It is important that $M$ is a matching, and not just any set of edges, since it allows the martingale analysis in \cite{BlikstadSVW25} to argue that the step size is bounded. Indeed, if an edge arrives, it can only affect the values $P_{fc}$ for up to two edges $f$ in a matching $M$. If $M$ was not a matching, the step size could potentially blow up by a factor $\Delta$, which is too much.

\begin{proof}[Proof of \cref{lem:matchings}]
    We note that initially $K_{MC}\idxz = |M|\cdot|C|\cdot \frac{1-\eps}{\Delta}$. Indeed, $K_{MC} = \sum_{e\in M,c\in C} P_{ec}$ is a sum of super-martingales, so itself is a supermartingale, and since the edges $e$ forms a matching it has just twice the step size. There are at most $2|M|$ endpoints of $M$, and each can have $\Delta$ incident edges making for a non-trivial step of the martingale. We refer to \cite{BlikstadSVW25} for a proof.
\begin{claim}[{\cite[See proof of Lemma 4.17]{BlikstadSVW25}}]
    $K_{MC}\idxt$ is a super-martingale with step size $24A$, and at most $2|M|\Delta$ non-trivial steps.
\end{claim}
    Thus we can define a potential $\Phi_{K_{MC}}$ using \cref{lem:potential-of-martingale} with parameter $\lambda = \frac{\eps|M|\cdot|C|}{2\Delta}$ and get:
    \begin{itemize}
        \item If $\Phi_{K_{MC}}(t) < 1$, then $K_{MC}\idxt \le \frac{|M|\cdot|C|}{\Delta}(1-\frac{\eps}{2})$.
        \item Initially $\Phi_{K_{MC}}(0) < 2^{-95\Delta}n^{-95}$ (see calculation in \cite[Lemma~4.17]{BlikstadSVW25}).
        \item $\Phi_{K_{MC}}$ is non-increasing in expectation.
        \item $\Phi_{K_{MC}}$ is updated only when some edge $f$ arrives which shares an endpoint with some edge $e\in M$. Indeed,for $e\in M$, the values $P_{ec}$ is updated only when an edge $f$ incident to $e$ arrives.
    \end{itemize}

    We now need to bound $K_{MC}$ from below too, with another potential function. Let us assume that no endpoint of $M$ is ever bad, and that \labelcref{inv:fewbadcol} holds. Indeed, if either of these is not the case, we do not need to prove any concentration for $K_{MC}$ from below.

    Since $K_{MC}$ is a super-martingale, we can write $K_{MC}\idxt = L_{MC}\idxt + K_{MC}\idxz + \text{drift}\idxt$ where $L_{MC}$ is a martingale centered at $0$, and $\text{drift}\idxt$ is some downward drift incured up to time $t$. In particular, we can set 
    $L_{MC}\idxt= \sum_{e\in M}\sum_{c\in C} Y_{ec}$. Since there are no bad endpoints of $M$, the only downwards drift can come from skipping scaling up some $P_{ec}$ due to $P_{ec} > A$ in Line~\ref{line:obliv:scaleup} (this is exactly what is formalized with the martingalse $Y$ that depends on per time-wise non-capped version of the $P_{ec}$'s). By invariant~\labelcref{inv:fewbadcol}, for each edge $e$ and any timestep $t$, there can be at most $2\eps^5\Delta$ many such $c$ leading to the downward drift. For each edge $f$, there are up to $2\Delta$ times the values $P_{fc}$ are updated (for each incident edge $e_{t}$), and each time the downward drift is thus bounded by
    \begin{align*}
    \sum_{c\in C, P\idx{t-1}_{fc}>A}(\overline{P}_{fc}\idxt - P_{fc}\idxt)
        &=
        \sum_{c\in C, P\idx{t-1}_{fc}>A}(\frac{P_{fc}\idxt}{1-P_{e_{t}c}\idxt} - P_{fc}\idxt)
    \\
        &=
        \sum_{c\in C, P\idx{t-1}_{fc}>A}(\frac{P_{fc}\idxt P_{e_{t}c}\idxt}{1-P_{e_{t}c}\idxt})
        \\
        &\le
        \sum_{c\in C, P\idx{t-1}_{fc}>A}\frac{2A\cdot 2A}{1-2A}
        \\
        &\le
        |\{c\in C : P\idx{t-1}_{fc}>A\}| \cdot 5 A^2
        \le 10 A^2 \eps^5 \Delta
    \end{align*}
        (See also \cite[Lemma~4.15 and Section~3.3]{BlikstadSVW25} for more detailed calculations). In total, over the $2\Delta$ arrivals, this means that $P_{eC}$ is a martingale minus a downward drift of $20\eps^5A^2\Delta^2 = 20c_{A}^2 \eps$ (which, intuitively, is almost negliable). For $K_{MC}$, the drift down can be summed over the $|M|$ edges for a total of $20c_{A}^2 \cdot |M|$. Hence, if no endpoints of $M$ are bad, we have $K_{MC}\idxt \ge L_{MC}\idxt + K_{MC}\idxz - 20c_{A}^2\eps \cdot |M|$. Since we are aiming to prove that $K_{MC} \ge |M|\cdot|C|/\Delta - 35c_{A}^2\eps |M|$, and initially $K_{MC}\idxz = (1-\eps)|M|\cdot |C|/\Delta$, it suffices to show that the martingale
$L_{MC}\idxt$ falls down at most $-\frac{|M|\cdot|C|}{\Delta}(1-\frac{\eps}{2})$.

        Thus we define another potential $\Phi_{-L_{MC}}$ using \cref{lem:potential-of-martingale} on the (super-martingale) $-L_{MC}$ (which has the same step size and number of non-trivial steps as $K_{MC}$) with parameter $\lambda = \frac{\eps|M|\cdot|C|}{2\Delta}$ and get:
    \begin{itemize}
        \item If $\Phi_{-L_{MC}}(t) < 1$, then $L_{MC}\idxt \ge -\frac{|M|\cdot|C|}{\Delta}(1-\frac{\eps}{2})$.
        \item Initially $\Phi_{-L_{MC}}(0) < 2^{-95\Delta}n^{-95}$.
        \item $\Phi_{-L_{MC}}$ is non-increasing in expectation.
        \item $\Phi_{-L_{MC}}$ is updated only when some edge $f$ arrives which shares an endpoint with some edge $e\in M$.
    \end{itemize}

        We then simply set $\Phi_{M,C} = \Phi_{K_{MC}} + \Phi_{L_{MC}}$, and this potential function will guarantee concentration for $K_{MC}$ both from above and below, proving \cref{lem:matchings}.
\end{proof}

\subsubsection{Bad Vertex Property.}
        Finally we show that invariant~\labelcref{inv:badvertexprop} holds by showing \cref{lem:badvertexprop}. The proof is very similar to the ``few bad vertices'' invariant above.

\begin{align*}
    \Phi_{w, \text{bad vertex prop}} &= \sum_{C\subseteq [\Delta], |\Delta| = 2c_{K}\eps \Delta} \ \sum_{U\subseteq N(w), |U| = \eps \Delta} \Phi_{w, C,U}\\
    \Phi_{w, C,U} &= \Phi_{X_{C,U}} + \sum_{M_{i}\in \mathcal{M}(w,U)} \Phi_{w, C, M_{i}}\\
    \Phi_{w, C, M_{i}} &= \sum_{M \subseteq M_{i}, |M| = \eps^3 \Delta} \Phi_{M,C}
\end{align*}
We will define $\Phi_{X_{C,U}}$ below, and recal the definition of $\Phi_{M,C}$ from \cref{lem:matchings}.

Suppose for the sake of a contradiction that invariant~\labelcref{inv:badvertexprop} does not hold. That is, that vertex $w$ turned bad and that (at least) $\eps \Delta$ of the arriving edges (call them \emph{problematic}) incident to $w$ had no availible color in $[\Delta]$ (i.e., $Z_{e} = 0$ for those edges). Say $U\subseteq N(v)$ are the $|U| = \eps\Delta$ many neighbors of $w$ these problematic edges connected $w$ to. Since $w$ is \emph{bad}, we know that at least $c_{K}\eps\Delta$ edges incident to $w$ was forwarded to greedy, and since $w$ has degree at most $\Delta$, there must exist some subset $C\subseteq [\Delta]$ of $|C| = c_{K}\eps \Delta$ colors that will be free at $w$ at the end of the algorithm. In particular, none of these colors $C$ could be free for $U$ at the end of the algorithm (since then when edge $(w,u)$ for $u\in U$ arrives, that color was free for both $w$ and $u$, and would contradict that the edge was problematic).

This motivates the following definition for each $C$ and $U$ of the sizes above:
\begin{align*}
X'_{C,U} = \sum_{u\in U}(\text{\# incident edges at $u$ that are colored from $C$})
\end{align*}
(note that some edges might be counted twice in the above sum if they are incident to two vertices in $U$). We do not count the edges $e$ connecting $u$ to $w$ in the definiton of $X'_{C,U}$.

Note that if $X'_{C,U} < |U|\cdot |C|$, then there must be at least one vertex $u$ in $U$ that has some free color $c\in C$. Thus if we can show that $X'_{C,U} < |U|\cdot \Delta$ for all $C\subseteq[\Delta]$ and $U\subseteq N(v)$ of sizes $|C| = c_{K}\eps \Delta$ and $|U| = \eps \Delta$, we would have shown that invariant~\labelcref{inv:badvertexprop} holds for vertex $w$. This is thus out goal, and motivates our potential.

    From the potential $\Phi_{M,C}$ in the definition of $\Phi_{w,\text{bad vertex prop}}$ coming from \cref{lem:matchings}, we have the following for each matching $M_{i}\in \mathcal{M}(w,U)$:
\begin{itemize}
    \item $\Phi_{w,C,M_{i}}(t) < 1$ implies that at most $\eps^3\Delta$ edges $e\in M_{i}$ have $Z_{eC}\idxt > (1-\frac{\eps}{2})\cdot\frac{|C|}{\Delta}$ (since otherwise a submatching $M$ of size $\eps^3\Delta$ would violate the upper bound in \cref{lem:matchings}).
    \item Initially $\Phi_{w,C,M_{i}}(0) = 2^{-80\Delta}n^{-80}$ (as there are at most $2^{\Delta}$ submatchings $M\subseteq M_{i}$).
    \item $\Phi_{w,C,M_{i}}$ is non-increasing in expectation.
    \item $\Phi_{w,C,M_{i}}$ can be computed from the $3$-hop neighborhood of $w$ (since each edge $M\in M_{i}$ is in the $2$-hop neighborhood of $w$).
\end{itemize}
Thus, the total contribution of $\sum_{C}\sum_{M_{i}} \Phi_{w,C,M_{i}}$ to $\Phi_{w,\text{bad vertex prop}}$ is at most $2^{-70\Delta}n^{-70}$, as there are $\le 2\Delta$ matchings $M_{i}\in \mathcal{M}(w,U)$, and at most $2^\Delta$ subsets of colors $C\subseteq [\Delta]$.

In total over all the $\le 2\Delta$ matchings $M_{i}$, there can be at most
$\eps^3\Delta \cdot 2\Delta= 2\eps^3\Delta^2$ edges $e$ with $Z_{eC}\idxt > (1-\frac{\eps}{2})\frac{|C|}{\Delta}$ at the time of arrivial of $e$.
Worst case, each of these steal a color from $C$, and thus contributes a total of
$4\eps^3\Delta^2$ to $X'_{C,U}$ (remember that each edge can contribute twice if it connects two vertices in $U$).

Morover, there are at most $(\alpha+\eps\alpha)\Delta^2$ edges in $G[U]$ that are incident to a bad vertex: at most $\alpha$ vertices in $U$ are bad (since they all neighbor $w$), and at most $|U|\cdot \alpha\Delta = \eps\alpha \Delta^2$ neighbors of $U$ are bad. This is because we assume ``few bad neighbors'' invariant~\labelcref{inv:fewbadnei} holds. Each of these edges (with some bad endpoint) contribute at most $2$ to $X'_{C,U}$ (in the worst case when they are connected to two vertices in $U$, and are colored from $C$). In total this is at most $2(\alpha+\eps\alpha)\Delta^2$.

Finally, what remains is to consider the edges $e_{t}$ incident to $U$ that have no bad endpoints, and for which $Z_{e_{t}C}\le (1-\frac{\eps}{2})\frac{|C|}{\Delta}$. This is the ``common case'', most edges will be like this. A simple Chernoff bound (with coupling) in \cite{BlikstadSVW25} shows that they are well-concentrated around their mean, and thus we define a potential to dereandomize this Chernoff bound.

We define $\Phi_{X_{C,U}}$ similarly as $\Phi_{H_{U}}$ was for the ``few bad neighbors'' potential. In particular, define $X_{C,U}\idxz = 0$, and whenever an edge $e_{t}$ incident to $U$ arrives, such that both endpoints of $e_{t}$ are good and $Z_{e_{t}C}\le (1-\frac{\eps}{2})\frac{|C|}{\Delta}$, then we update $X_{C,U}$ as follows\footnote{if some other edge arrives, we leave it unchanged: $X_{C,U}\idxtp \gets X_{C,U}\idxt$.}:

\begin{align*}
  X_{C,U}\idxtp \gets X_{C,U}\idxt - s_{e_{t}}(1-\tfrac{\eps}{2})\tfrac{|C|}{|\Delta|} + s_{e_{t}}\cdot 1[\text{$e_{t}$ is colored from $C$}],
\end{align*}
where $s_{e_{t}}\in \{1,2\}$ is the number of vertices in $U$ that $e_{t}$ touches.
We see that $X_{C,U}$ is a supermartingale, since the probability of $e_{t}$ being colored $C$ is $Z_{e_{t}C}$, and it has step size $2$. It has at most $|U|\cdot\Delta = \eps\Delta^2$ (non-trivial) steps.

Thus $\Phi_{X_{C,U}}$ is defined as a martingale-controlling potential using \cref{lem:potential-of-martingale} with $\lambda = 2\eps^3\Delta^2(1-\frac{\eps}{2})$, and we obtain:
\begin{itemize}
  \item $\Phi_{X_{C,U}}(t) < 1$ implies that $X_{C,U}\idxt\le \lambda$.
  \item $\Phi_{X_{C,U}}(0) \le 2^{-110\Delta}n^{-110}$ (see calculation in \cite[Proof of Lemma~4.19]{BlikstadSVW25}).
  \item $\Phi_{X_{C,U}}(t)$ can be determined by inspecting only arrived edges incident to $U$ (and hence in the $2$-hop neighborhood of $w$).
  \item $\Phi_{X_{C,U}}$ is non-increasing in expectation.
\end{itemize}

Since $X_{C,U}\le \lambda$ at all times, we get that this contributes at most
the ``expectation'' $|U|\cdot \Delta \cdot (1-\frac{\eps}{2})\frac{|C|}{\Delta}$ plus deviation $\lambda$ to $X'_{C,U}$.
Indeed, this is at most $(2c_{K}+2\eps)\cdot \eps^2\Delta^2\cdot(1-\frac{\eps}{2})$ contribution to $X'_{C,U}$.

If we finally sum up the contribution to $X'_{C,U}$ from all cases above, we see that it is strictly smaller than $2c_{K}\eps^2\Delta^2 = |U|\cdot |C|$ (see calculation in \cite[Proof of Lemma~4.19]{BlikstadSVW25}). Thus we have showed that $\Phi_{w, \text{bad vertex prob}} < 1$ implies that invariant~\labelcref{inv:badvertexprop} holds (we also assumed in our proof that invariant~\labelcref{inv:fewbadnei} held). We also see that $\Phi_{w,\text{bad vertex prop}}(0) \le 2^{-50\Delta}n^{-50}$, and that it is the sum of martingale-controlling potentials (so non-increasing expectation). This proves \cref{lem:badvertexprop}.

\newpage
\section{Lower Bounds in \SLOCALedge{} and \LOCAL}\label{sec:LB}


In this section, we show some simple lower bounds for edge coloring in the \SLOCALedge{} model.
Our first lower bound shows that locality $3$ is necessary for our $(1+\eps)\Delta$ deterministic edge coloring algorithm (\cref{thm:slocal-det}).
Our second lower bound shows that our randomized algorithm with locality $1$ (\cref{thm:slocal-rand}) does indeed need the assumption that $\Delta = \Omega(\sqrt{\log n})$.

\ESLOCALLB*

The deterministic construction of \Cref{thm:LB_det} also works for smaller $\Delta$.

\begin{remark}
    Our lower bounds assume that there are no unique identifiers on the vertices or edges.\footnote{In \SLOCAL{} or \SLOCALedge{}, this is a reasonable assumption -- even for deterministic algorithms -- as the arrival order can be used for symmetry breaking instead. Our lower bound also hold if there are unique identifiers, but that these can only be compared for equality with each other (that is, the algorithm is not allowed to collect some identifiers, sort them, and make decisions based on their ordering).} If there are unique identifiers, the same lower bound holds for localities $\le 1$ instead.
\end{remark}

\paragraph{The lower bound construction.} The construction is similar to that of the lower bounds for online algorithms by Bar-Noy, Motwani, and Naor~\cite{BMN92}. The construction consists of first letting $\Delta$ many stars of degree $\Delta-1$ arrive, and then connecting the root of all those stars with a new node. See \cref{fig:lowerbound}. If the locality of the algorithm is at most $2$, decisions taken to color the stars must be independent (since they are far enough away in the final graph).

\begin{figure}[!ht]
    \includegraphics[page=2]{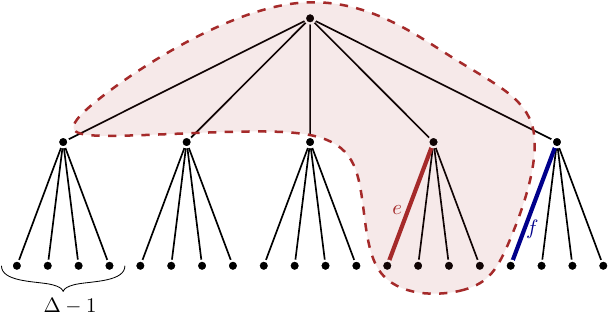}
    \centering
    \caption{The lower bound construction. Edges in the bottom layer arrive first. The $2$-hop of an edge $e$ in the bottom layer will not see any of the other stars. Thus, the viewpoints of edges $e$ and $f$ at the bottom looks identical, and decisions for coloring, say edge $e$, only depends on colors used on edges already arrived at the star containing $e$.}
    \label{fig:lowerbound}
\end{figure}

\paragraph{Deterministic algorithms.}
A deterministic algorithm must thus color all the $\Delta$ stars identically (as, from the perspective of the \SLOCALedge{} algorithm they look identical). Thus there are $\Delta-1$ colors used to color all the stars, but then another $\Delta$ colors are needed to color the edges connecting the stars. This means the algorithm needs a total of $2\Delta-1$ colors. Here $n \approx \Delta^2$, so $\Delta \approx \sqrt{n}$.

\paragraph{Randomized algorithms.}
A randomized algorithm might use randomization to color the stars in different ways.
Still, since the stars look identical to each other for any \SLOCALedge{} algorithm of locality at most $2$, the colorings of the stars must be drawn identically and independently from the same distribution $D$ (dependent on what randomized algorithm we use). Suppose the algorithm uses at most $2\Delta-2$ colors, for the sake of a contradiction. This means that there are $\binom{2\Delta-2}{\Delta-1} \le 2^{2\Delta}$ possible ways to color the stars. One of these ways must thus occur with probability $\ge 2^{-2\Delta}$ in the distribution $D$. Now, the probability that \emph{all} the $\Delta$ stars (independently) get that same coloring is at least $(2^{-2\Delta})^{\Delta} = 2^{-2\Delta^2}$.
If that happens, we are in a similar case as the deterministic algorithms: all the stars are colored with the same $\Delta-1$ colors, so we need another $\Delta$ fresh colors for the edges connecting the stars; for a total of $2\Delta-1$ colors (which is our desired contradiction). Thus, with probability $2^{-2\Delta^2}$, the randomized algorithm must use $2\Delta-1$ colors on the construction in \cref{fig:lowerbound}. If we repeat this construction $2^{2\Delta^2}$ times (in total we are using $n = \Delta^2 \cdot 2^{\Delta^2}$ nodes, so $\Delta \approx \sqrt{\log n}$), the randomized algorithm has constant probability of using $\ge 2\Delta-1$ colors.

\begin{theorem}
\label{thm:LOCALLBDelta}
For any constant integer $k>0$ and any constant $c>0$, there is no $o(\Delta^{k}/\log \Delta)+O(\log^* n)$-round \LOCAL algorithm for edge coloring graphs with maximum degree $\Delta\geq c\log^{1/k} n$ with fewer than $2\Delta-1$ colors. 
\end{theorem}
\begin{proof}
\cite{CHLPU18} showed that any deterministic $2\Delta-2$ edge coloring algorithm in the \LOCAL model requires $\Omega(\log_{\Delta}n)$ rounds. The non-existing algorithm of the theorem would immediately contradict that lower bound. 
\end{proof}

\newpage
\section{\LOCAL Edge Coloring}\label{sec:local}

In this section we prove \Cref{thm:LOCALmainNoND}, showing that we can efficiently compute a $(1+\eps)\Delta+O(\sqrt{\log n})$ coloring in the \LOCAL model. 

As outlined in \Cref{sec:technical_overview}, we obtain our results in two steps: we show that the online coloring algorithm is an \SLOCALedge algorithm with low locality, and we show that any \SLOCALedge algorithm can be turned into a \LOCAL algorithm with the use of an appropriate schedule, yielding our final distributed edge coloring algorithms. 
We have show the former in \Cref{sec:online_locality}, and this section focuses on the latter.


\paragraph{Scheduling.}
Ghaffari, Kuhn, and Maus~\cite{GhaffariKM17} showed that any \SLOCAL algorithm is a \LOCAL algorithm, with $\poly\log n$ multiplicative overhead. They do this via network decompositions. 
We give an alternative reduction from \SLOCALedge to \LOCAL using coloring, which gives faster results in our setting. 

In an \SLOCALedge($\ell$) algorithm, one edge is processed at a time, using the neighborhood within distance $\ell$. We observe that edges that are at least distance $\ell$ apart can be processed simultaneously if the order for the \SLOCALedge algorithm is chosen appropriately. This means that if we find a schedule in the form of an edge partition $E_1\cup E_2 \cup \cdot \cup E_T$,  such that each two edges in $E_i$ are at each at least at distance $\ell$, we can find a suitable order to simulate the \SLOCALedge algorithm with a \LOCAL algorithm that runs in time $O(\ell \cdot T)$, by treating each edge as an arrival in the online algorithm.  

\paragraph{Randomization.}
If the online algorithm is randomized then the \SLOCALedge algorithm is also randomized and succeeds with the same probability. Hereto, it is sufficient that the online algorithm works against an oblivious adversary, since the online algorithm is run in the \SLOCALedge model according to the precomputed schedule. This means that both the graph and the arrival order are fixed beforehand, thus can be seen as an oblivious adversary.

\subsection{Running \SLOCALedge in \LOCAL}\label{sec:ND_def}

 We create a schedule where edges can be processed simultaneously if they are at least $\ell$ hops apart, see \Cref{fig:path_dist_3_col} for an example where $\ell=3$. In other words, we need a distance-$\ell$ edge coloring: an edge coloring of $G$ such that edges within distance $\ell$ have different colors. This can also be seen as a vertex coloring of the $\ell$-th power of the line graph $L(g)$, i.e., of $L(G)^\ell$. 

 \begin{figure}[!ht]
    \includegraphics[page=1]{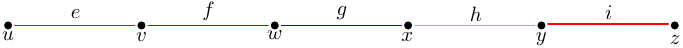}
    \centering
    \caption{A distance-$3$ coloring. Note that $e$ and $i$ can have the same color, while each other edge has a different color. In a distance-$1$ coloring (also \emph{edge coloring}): two colors would suffice: $g$ could also be red and $f$ and $h$ could both be green.}
    \label{fig:path_dist_3_col}
\end{figure}

\begin{lemma}\label{lm:SLOCAL_to_LOCAL_coloring}
    Given an \SLOCALedge algorithm with locality $\ell$, there exists a \LOCAL algorithm for the same problem that takes 
    $O(\ell\cdot\Delta^\ell+\log^*n)$ rounds.
\end{lemma}
\begin{proof}
    First, we compute a distance-$\ell$ edge coloring with $O(\Delta^\ell)$ colors on $G$. 

    \smallskip
    
    \textbf{Computing a distance-$\ell$ edge coloring.}
    We do this by computing a vertex coloring of $L(G)^\ell$, the $\ell$-th power of the line graph.  We know that we can compute an $O(\Delta')$ vertex coloring in $O(\Delta'+\log^* n)$ rounds~\cite{BarenboimEK14} on a graph of maximum degree $\Delta'$.
    Filling in that the maximum degree of $L(G)^\ell$ is $\Delta'=\Delta^\ell$, gives $O(\Delta^\ell+\log^*n)$ rounds of computation \emph{on $L(G)^\ell$}. However, computing on the $\ell$-th power graph (for constant $\ell$) and on the line graph only incur constant factors overhead in the \LOCAL model. 
    So we find a distance-$\ell$ edge coloring with $O(\Delta^\ell)$ colors in $O(\Delta^\ell+\log^*n)$ rounds. 

    \smallskip
    
    \textbf{Using the distance-$\ell$ edge coloring.}
    Next, we use this distance-$\ell$ edge coloring with $O(\Delta^\ell)$ colors.
    If edges $e$ and $e'$ have different colors, they have distance at least $\ell$, hence can be processed simultaneously by the \SLOCALedge algorithm. This can be performed in $\ell$ rounds by either of the endpoints, without loss of generality, by the node with largest ID.
    
    We execute the schedule by processing all edges with colors $1$, then $2$, etc., until $O(\Delta^\ell)$. This takes $(\ell\cdot \Delta^\ell)$ rounds in total. 
\end{proof}

We get the following corollary for complexity classes. Such a result was already known for the \SLOCAL model~\cite{GhaffariKM17}. Blackboxing these results would give the \SLOCALedge version with $\ell\leftarrow \ell+2$. We note that $\ell$ appears as a power in the second result, which is why we have to be more careful here for better results. 

\CorESLOCAL*
\begin{proof}
    %
    This follows immediately from \Cref{lm:SLOCAL_to_LOCAL_coloring}. 
\end{proof}






\subsection{Proof of Theorem~\ref{thm:LOCALmainNoND}}
Next, we provide a proof of \Cref{thm:LOCALmainNoND}. We do this in two steps. First, we show \Cref{thm:LOCALmainNoNDNoDS}, which has a polynomial dependency on $\Delta$. Then we use degree splitting to reduce to the case of small $\Delta$. 


Naively, our deterministic \SLOCALedge{}(5) algorithm can be scheduled using a distance-$5$ edge coloring. Indeed, if two (or more) edges $e$ and $f$ are of distance more than $5$, then they can be scheduled at the same time as the \SLOCALedge{} algorithm will not make them read eachother (although, they migh have read-only access to the same edges in between, but that is okay). This would require a schedule of $\Theta(\Delta^5)$ rounds, but below we argue that for this particular algorithm we can do it in $O(\Delta^4)$ rounds by computing a shorter schedule.

\LOCALNoNDNoDS*
\begin{proof}

    To obtain the runtime stated in the theorem, we first compute a edge coloring $C'$ of $G$ with $2\Delta-1\leq 2\Delta$ colors in $O(\Delta+\log^* n)$ rounds with the algorithm of \cite{BarenboimEK14}. We run the same \SLOCALedge{} algorithm as in \Cref{sec:online_locality}, with a minor change: all the canonical matchings $\mathcal{M}(w,U)$ will be taken from the $2\Delta$-coloring $C'$. This will increase the number of such matchings from $\Delta+1$ to $2\Delta$, but everywhere in the proof we rounded the $(\Delta+1)$ term to $2\Delta$ anyway. 
    Then, we define the conflict graph of edges that cannot be scheduled simultaneously by this algorithm.

    \paragraph{Conflict graph.} Define the conflict graph $G_{\text{conflict}}$ where the vertices are the edges of the original graph, and two vertices (read: edges) are connected if and only if they cannot be scheduled simultaneously. Equivalently, for our potential-function based algorithm, this means that two edges are connected if they affect the same part of the potential function.

Any independent set in the conflict graph can be scheduled simultaneously.  Naively, the conflict graph of our \SLOCALedge{} algorithm has degree $\Delta' = \Theta(\Delta^5)$, as it has locality $5$. Hence any $O(\Delta')$ (vertex)-coloring of $G_{\text{conflict}}$ can be used as a schedule of length $\Delta^5$. We show that we can decrease the maximum degree $\Delta'$ of the conflict graph to $O(\Delta^4)$ instead, so that we can get a shorter schedule.

    \begin{claim}\label{claim:conflict_graph_size}
    $G_{\text{conflict}}$ has maximum degree $O(\Delta^4)$.
    \end{claim}
    \begin{proof}
    Consider some edge $e = (u,v)$. We look at what edges are in conflict with $e$. Most of our potential functions will be centered at some vertex $w$ and be affected by edges in the $2$-hop of $w$. For those potentials, $e$ will be in conflict with all other edges of distance at most $3$ away from $e$. The only part of the potentials $\Phi_{w}$ that actually require the $3$-hop from $w$ (and would naively require all edges within distance $5$ from $e$ to be conflicts) are the potentials $\Phi_{M,C}$ from \cref{lem:matchings}, where $M$ are some matchings of edges within the $2$-hop of some vertex $w$. Edge $e$ affects $\Phi_{M,C}$ if $e$ shares an endpoint with some edge in $M$.

    This means that $e$ is only in conflict with an edge $f$, if there exists a path of edges $e,a,b,c,d,f$ where $a$ and $d$ belong to the same matching (i.e., are colored the same color in $C'$). In particular, there are $2\Delta$ choices for $a$, $\Delta$ choices for $b$, $\Delta$ for $c$, \textbf{only at most one choice for $d$} (given the choice of $a$), and $\Delta$ choices for $f$. In total this contributes $2\Delta^4$ to the conflict-degree of $e$, and not $\Theta(\Delta^5)$: we save a factor of $\Delta$ from the construction of the potentials all being based on a consistent partition into matchings $C'$.
\renewcommand{\qed}{\ensuremath{\hfill\blacksquare}}
\end{proof}
\renewcommand{\qed}{\hfill \ensuremath{\Box}}

     $G_{\text{conflict}}$ has maximum degree $O(\Delta^4)$ and any two non-adjacent vertices of $G_{\text{conflict}}$ (aka edges of $G$) can be  be processed simultaneously in the \SLOCALedge algorithm.  We can compute the necessary schedule aka vertex coloring of $G_{\text{conflict}}$ in $O(\Delta^4+\log^* n)$ rounds with the algorithm by \cite{BarenboimEK14}. Iterating through the schedule and deciding on the color of each edge takes $O(\Delta^4+\log^* n)$ rounds and computes the desired $(1+\eps)\Delta$-edge coloring. 
\end{proof}

\paragraph{Degree splitting.}
\Cref{thm:LOCALmainNoNDNoDS} uses $\poly(\Delta)$ rounds. To make this more efficient, we combine it with the standard technique of \emph{degree splitting}. 
The degree splitting problem seeks a partitioning of the graph edges $E$ into two parts so that the partition looks almost balanced around each node. Concretely, we should color each edge red or blue such that for each node, the difference between its number of red and blue edges is at most some small discrepancy value $\kappa$. In other words, we want an assignment $q \colon E \to \{+1, -1\}$ such that for each node $v\in V$, we have
\begin{equation*}
    \left|\sum_{e\in E(v)}q(e)\right| \le \kappa,
\end{equation*}
where $E(v)$ denotes the edges incident on $v$. We want $\kappa$ to be as small as possible.

\begin{lemma}[\cite{GhaffariHKMSU20}]\label{lm:degree_splitting_LOCAL}
    For every $\eta> 0$, there is a deterministic 
    $\tilde O(\eta^{-1}\cdot \log n)$
round \LOCAL algorithm for computing degree splitting with discrepancy at each node at most $\eta \cdot d(v) +O(1)$. 
\end{lemma}

Using degree splitting for coloring has been done before, see, e.g.,~\cite{GhaffariHKMSU20,BMNSU25}. We include a proof here for completeness.

\thmLOCALmainNoND*
\begin{proof}
    For $\Delta\le \Theta(\sqrt{\log n})$, this follows from, e.g., \cite{BEG17}, which provides $2\Delta-1$ coloring in $O(\Delta+\log^*n)$ rounds. 

    The remainder of the proof is dedicated to $\Delta\ge \Theta(\sqrt{\log n})$.
    Our goal is to edge-partition $G$ into  subgraphs $G_1, G_2, \dots$, such that each subgraph $G_i$ has low maximum degree $\Delta(G_i)$ and the sum of their maximum degrees $\sum_i \Delta(G_i)$ is bounded by $(1+\eps)\Delta$. If this is the cases, then we can color each subgraph $G_i$ with $(1+\eps)\Delta(G_i)$ colors, using $\sum_i (1+\eps)\Delta(G_i)\le (1+\eps)^2 \Delta$ colors in total. Setting $\eps \leftarrow\eps/3$ then gives the result. 

    Be obtain this partition by iteratively applying the degree splitting lemma \Cref{lm:degree_splitting_LOCAL} with $\eta := \tfrac{\eps}{2\log(\Delta/\Delta')}$ until we reach parts that have maximum degree $\Delta' :=\Theta(\sqrt{\log n})$. We will show that $k=\log(\Delta/\Delta')$ recursive iteration suffice for this.     
    If the maximum degree of each part after $i$ iterations is bounded by $\Delta_i$, then by \Cref{lm:degree_splitting_LOCAL} $\Delta_{i+1}\le \tfrac{1}{2}(\Delta_i+\eta\Delta_i+\gamma)$, for some constant $\gamma>0$.
    
    Now we claim by induction that 
    \begin{equation}\label{eq:deg_splitting_induction}
        \Delta_i \le \left( \frac{1+\eta}{2}\right)^i\Delta  +\frac{\gamma}{2}\sum_{j=0}^{i-1}\left( \frac{1+\eta}{2}\right)^j.
    \end{equation}
    The base case is immediate: $\Delta_1 \le (1+\eta)\Delta+\tfrac{\gamma}{2}$ by \Cref{lm:degree_splitting_LOCAL}. Now in the induction step, we see 
    \begin{align*}
        \Delta_{i+1} \le \frac{1}{2}(\Delta_i+\eta\Delta_i+\gamma) &\le \frac{1}{2}(1+\eta)\left(
        \left( \frac{1+\eta}{2}\right)^i\Delta +\frac{\gamma}{2}\sum_{j=0}^{i-1}\left( \frac{1+\eta}{2}\right)^j\right)+\frac{\gamma}{2} \\ 
        &= \left( \frac{1+\eta}{2}\right)^{i+1}\Delta +\frac{\gamma}{2}\sum_{j=0}^{i}\left( \frac{1+\eta}{2}\right)^j,
    \end{align*}
    which proves \eqref{eq:deg_splitting_induction}. Note that using the geometric sum we have that for $\eta \le 1-\gamma/5$
    \begin{equation*}
        \frac{\gamma}{2}\sum_{j=0}^{i-1}\left( \frac{1+\eta}{2}\right)^j= \frac{\gamma}{2}\frac{1-\left( \frac{1+\eta}{2}\right)^i}{1- \frac{1+\eta}{2}} \le \frac{\gamma}{1-\eta} \le 5. 
    \end{equation*}
    Using this, we see that 
    \eqref{eq:deg_splitting_induction} in particular implies that
    \begin{equation}\label{eq:deg_splitting_bound_on_Delta_k}
        \Delta_k \le \left( \frac{1+\eta}{2}\right)^k\Delta +5 \le 2^{-k}\exp(\eta\cdot k)\Delta + 5 \le (1+\eps)\Delta'+5 \le \Theta(\Delta').
    \end{equation}
    So indeed $k=\log(\Delta/\Delta')$ iterations suffice to make all parts have maximum degree $\Theta(\Delta')$. 

    We color each subgraph now using \Cref{thm:LOCALmainNoNDNoDS} with $(1+\eps)\Delta_k$ colors. 
    Using \eqref{eq:deg_splitting_bound_on_Delta_k}, we have in total
    \begin{equation*}
        (1+\eps)2^k \Delta_k \le (1+\eps) (1+\eps)\Delta +5\cdot \Delta/\Delta'. 
    \end{equation*}
    We note that $\Delta' = \Theta(\sqrt{\log n}\ge 1/\eps)$, so $5\cdot \Delta/\Delta'\le 5\eps\Delta$. 
    Now setting $\eps \leftarrow \eps/10$ gives that we used at most $(1+\eps)\Delta$ colors. 

    \smallskip
    
   \textbf{Round complexity.}
   We run at most $k=\log(\Delta/\Delta')$ iterations of degree splitting with $\eta= \tfrac{\eps}{2\log(\Delta/\Delta')}$. Each such iteration takes 
   $\tilde O(\eta^{-1}\log n)$ rounds, by \Cref{lm:degree_splitting_LOCAL}. So in total we use $\tilde O(k\cdot \eta^{-1}\log n)=\tilde O(\log^2\Delta\log n)$ rounds for degree splitting. Then on each resulting graph, we have maximum degree $\Delta' =\Theta(\sqrt{\log n})$, which we color in $O(\Delta'^3)=O(\log^{3/2} n)$ rounds using \Cref{thm:LOCALmainNoNDNoDS}. 
\end{proof}

\newpage
\section{\CONGEST Edge Coloring}
\label{sec:congest}
In this section, we prove \Cref{thm:CONGESTmain}, showing that we can efficiently compute a $(1+\eps)\Delta+O(\sqrt{\log n})$ coloring in the \CONGEST model. The structure of our proofs will be analogous to \Cref{sec:local}. However, we cannot rely on \Cref{cor:complexity} and we need to whitebox several procedures in order to  show that we can implement them under the additional bandwidth constraints. 


\subsection{An Efficient Schedule}\label{sec:CONGEST_schedule}

The following lemma is a simplification of \cite[Thereom 1.1]{M21} with $m=\Delta^k$, $d=0$ and $k=1$. 
\begin{lemma}\label{lm:easy_coloring}
    There exists a $\Delta$-round \CONGEST algorithm that, given a constant $k\ge 1$ and an undirected graph $G=(V,E)$, together with an $O(\Delta^k)$ vertex coloring, outputs an $O(\Delta)$ vertex coloring. 

    The algorithm consists of internally computing some order of $O(\Delta)$ colors. Each round consists of trying one color. If there are no conflicts, it is adopted, otherwise it is rejected. 
\end{lemma}

We use this lemma to compute a schedule.

\begin{lemma}\label{lm:dist_l_coloring_conflict_graph}
    There exists a deterministic \CONGEST algorithm that computes an edge coloring  of the conflict graph with $O(\Delta^4)$ colors in $O\left(\Delta^5\log^*n+\Delta^{4}\cdot \ceil{\tfrac{\Delta^2\log \Delta}{\log n}}\right)$ rounds. 
\end{lemma}
\begin{proof}
The main idea is to apply \Cref{lm:easy_coloring} to the conflict graph. We know that $G_{\text{conflict}}\subseteq G^5$ and $\Delta(G_{\text{conflict}})=O(\Delta^4)$, see \Cref{claim:conflict_graph_size}. Hence an $O(\Delta(G_{\text{conflict}}))$ coloring of the edges gives the required schedule. 

To apply \Cref{lm:easy_coloring}, we first need a $\poly \Delta$ coloring of the conflict graph. Since our conflict graph is a subgraph of $L(G^5)$, we just color $L(G^5)$. This means that all edges within distance $5$ need to get distinct colors. We assign each edge to the endpoint with the highest ID. We now first color the vertices of the graph, such that vertices that are distance $\le 6$ apart get distinct colors. In other words, we color the vertices of $G^6$.  
Barenboim and Goldberg~\cite{BarenboimG24} give a deterministic \CONGEST algorithm for an $O(\Delta^\ell)$ vertex coloring of $G^{\ell}$ in $\tilde O(\ell\cdot \Delta^{\ell-1}\log^* n)$ rounds. In particular, in $O(\Delta^5\log^*n)$ rounds, this gives a vertex coloring of $G^6$ with $O(\Delta^6)$ colors. We turn this into an edge coloring as follows: each vertex $v$ with color $c_v$ considers the edges $e_1, \dots e_\Delta$ assigned to it. It now colors $e_i$ with colors $(c_v,i)$. This means that we obtain an intermediate coloring that colors $G_{\text{conflict}}$ with $\Delta\cdot O(\Delta^6)=O(\Delta^7)$ colors. In particular, we have a $\poly(\Delta_{\text{conflict}})$ coloring of the conflict graph $G_{\text{conflict}}$.

Now we are ready to apply \Cref{lm:easy_coloring}, which gives the required coloring. It takes $O(\Delta(G_{\text{conflict}}))=O(\Delta^4)$ rounds if each edge could directly communicate with each edge that is neighboring in the conflict graph. This is not the case. They are clearly a constant distance away, so what remains for us is to analyze the congestion of one color try.
\begin{itemize}
    \item Every node learns its $5$-hop neighborhood in $\Delta^5$ rounds. Then it internally constructs the part of the conflict graph it participates in, see \Cref{claim:conflict_graph_size}. 
    \item Each edge tries one color (since the algorithm is deterministic, both endpoints know this)
    \item It broadcasts this color for $2$ hops. If any such message meets a message from an edge that is neighboring in the conflict graph:
        \begin{itemize}
            \item If the colors are different: continue.
            \item If the colors are the same: abort and report. 
        \end{itemize}
\end{itemize}

See \Cref{fig:path_easy_coloring} for a picture of the situation. The edges $e$ and $f$ are within distance $5$: $d(e,f)=5$. During a color try, they need to check wether they have the same color. The nodes $v_2$ and $v_6$ will generate the color for the edges $e$ and $f$ respectively, they will broadcast it for $2$ hops. The node $v_4$ will receive both messages, and if $e$ and $f$ are neighboring in the conflict graph, it will compare them.

The congestion that appears due to this is $O(\Delta^2)$, since $v_2$ needs to send colors for $O(\Delta)$ edges, and $v_3$ has $O(\Delta)$ such neighbors.  Naively, this now takes $O(\Delta^2)$ rounds for each color try. However, we can leverage that we do not have big messages: each color try consists of some vertex IDs and a color. The former is written in $O(\log n)$ bits, the latter in $O(\log \Delta)$ bits. In fact, we do not need unique vertex IDs over the whole graph: they do just need to be unique within our $5$-hop neighborhood. Hence we use the vertex coloring with $O(\Delta^6)$ colors as IDs, now only needing $O(\log \Delta)$ bits for each ID. That means that the total congestion becomes $O(\Delta^2\log \Delta)$ bits. Since we have bandwith $\Theta(\log n)$ bits, we can transmit this in $O\left(\ceil{\tfrac{\Delta^2\log \Delta}{\log n}}\right)$ rounds. 


\begin{figure}[!ht]
        \includegraphics[page=1]{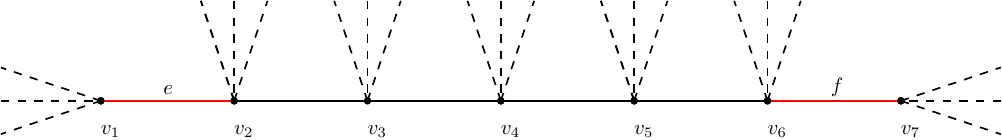}
        \centering
        \caption{The edges $e$ and $f$ are within distance $5$: $d(e,f)=5$. During a color try, they need to check wether they have the same color.}
        \label{fig:path_easy_coloring}
    \end{figure}

This gives a total of $O\left(\Delta^5\log^*n+\Delta^{4}\cdot \ceil{\tfrac{\Delta^2\log \Delta}{\log n}}\right)$ rounds. 
%
%
%
\end{proof}

\subsection{Main Result}\label{sec:CONGEST_main}

For our main result, we use degree splitting to make the maximum degree low. 
This is essentially known from prior work~ \cite{GhaffariHKMSU20,Nolin25}. We include more details in \Cref{sec:Congest_DS}. In particular, combining \Cref{lm:degree_spliting_CONGEST,lm:SO_CONGEST} gives the following result. 

\begin{lemma}\label{lm:degree_splitting_congest_logn}
    For every $\eta> 0$, there is a deterministic $\tilde O(\eta^{-1} \log n)$
round \CONGEST algorithm for computing degree splitting with discrepancy at each node at most $\eta \cdot d(v) +\gamma$, where $\gamma$ is some small, fixed constant. 
\end{lemma}

\thmCONGESTmain*
\begin{proof}
    For $\Delta\le \Theta(\sqrt{\log n})$, this follows from, e.g., \cite{BEG17}, which provides $2\Delta-1$ coloring in $O(\Delta+\log^*n)$ rounds. 

    The remainder of the proof is dedicated to $\Delta\ge \Theta(\sqrt{\log n})$.
    In the proof of \Cref{thm:LOCALmainNoND}, we showed that deterministic degree splitting reduces the problem to disjoint graphs of maximum degree $\Delta'=\Theta(\sqrt{\log n})$. In \CONGEST, we can compute deterministic degree splitting in $\tilde O(\log^2\Delta \log n)$ rounds by \Cref{lm:degree_splitting_congest_logn}. 
    Since the resulting graphs are disjoint, we can run \CONGEST algorithms on each of them without worrying about additional congestion. Formally, each message also contains a header indicating the subgraph that the message belongs to. The schedules executed in the disjoint subgraphs in parallel are of length $O(\Delta'^4)$ and can be computed with \Cref{lm:dist_l_coloring_conflict_graph} in $O\left(\Delta'^5\log^*n+\Delta'^{4}\cdot \ceil{\tfrac{\Delta'^2\log \Delta'}{\log n}}\right)= \tilde O(\log^{2.5}n)$ rounds.    
    Each schedule has $O(\Delta'^4)$ steps, where in each step an edge needs information from at most $5$ hops away (\Cref{thm:online_is_ESLOCAL}). The amount of information is at most $\poly(\Delta')$, so takes $O(\poly(\Delta'))$ rounds in total \CONGEST. Next, we investigate how much information this is exactly, and what the total congestion is. 

     The \SLOCALedge model, considers edge-arrivals. In the \LOCAL model, we can always run an algorithm `on an edge' with constant overhead: the necessary information is located at both endpoints. It can be transferred to one endpoint in $1$ round, where the computation can be performed, and the result can again be broadcasted in $1$ round. In the \CONGEST model, any information that needs to be known `on an edge' needs to be transmitted \emph{over} the edge, such that the relevant endpoint can make the decision. So our goal remains to show what needs to be known `on an edge'. 

    Firstly, every unprocessed edge needs to know its sampling probabilities. This can simply be updated whenever a neighboring edge is processed. This is at most $O(\Delta')$ messages, since there are $O(\Delta')$ possible colors in the distribution. However, the information does not need to propagate further. The bottleneck will be in the derandomization part of the algorithm. 

    By \Cref{lem:potential-gives-slocal}, when processing an edge $e=uv$ we need information from the $5$-hop neighborhood for derandomization. 
    
     From each unprocessed edge, we only need to know existence. That means that every node needs to know its $5$-hop neighborhood. Every edge has to forward $O(\Delta'^5)$ messages for this. Since this only needs to happen once, it is already dominated by the time needed to compute the schedule. 

    From each processed (arrived) edge $f$, we need to know
    \begin{enumerate}[(1)]
        \item the time of arrival $t$; \label{item:info_needed_time}
        \item the color assigned to it; \label{item:info_needed_color}
        \item the values $P_{fc}^{(t)}$ up to precision $\poly(\Delta')$; and \label{item:info_needed_P}
        \item the line in the algorithm $f$ was assigned a color. \label{item:info_needed_line}
    \end{enumerate}
    
     \ref{item:info_needed_time}, \ref{item:info_needed_color} and \ref{item:info_needed_line} clearly fit in a message of size $O(\log \Delta')$. The bottleneck is the $O(\Delta')$ values $P_{fc}^{(t)}$ (one for each color $c$) from \ref{item:info_needed_P}. Each such value needs to be given with precision $O(\poly \Delta')$, so it fits in $O(\log(\Delta')$ bits. 

     We note that we do not need this information for \emph{all} nodes within 5 hops, but only for those connected with the conflict graph: $O(\Delta'^4)$ many by \Cref{claim:conflict_graph_size}. 

    In total this means that we need to collect $O(\Delta'^5)$ messages of $O(\log n)$ bits `on' $e$. Since $e$ is processed within $u$ or $v$, this means we need to send $O(\Delta'^5)$ messages over $e$. Hence requesting this information upon the `arrival' of $e$ takes $O(\Delta'^5)$ time.     
    
    To get a faster algorithm, we do not `pull' this information, but rather `push' it. This means that whenever an edge $e'$ fixes its color, it informs its 5-hop neighborhood. Then the edge $e$ already has all the information it needs when it is being processed.

    Now let us consider how many rounds we need for such a push operation. Since we only push information $5$ hops, we only need to count congestion. Consider an edge $e=uv$ that needs to receive information from $f$ and $g$. That means that $f$ and $g$ are both neighboring $e$ in the conflict graph. In particular, it means that $e$ is within $5$ hops of $f$ and $g$, i.e., $d(e,f)\le 5$ and $d(e,g)\le 5$.
    
    Moreover, assume the shortest paths to $f$ and $g$ both go through $u$.
    We claim that $f$ and $g$ share at most one edge $a$ on their path to reach $u$.
    If not, their paths merge at some edge $b$, of distance at most $3$. See \Cref{fig:CONGEST_disjoint_paths} for an example.  That means that there are at most two edges between $f$ and $b$, and at most two edges between $g$ and $b$. This means that there is a path of $4$ edges between $f$ and $g$, so $d(f,g)\le 5$. Now since $e$ and $f$ are neighboring in the conflict graph, $a$ and $d$ need to be in the same matching $M$, see \Cref{claim:conflict_graph_size} for details. Similarly, $a$ and $d'$ need to be in the matching $M$. This means that $d(f,g)\le 5$ and $d$ and $d'$ are both in $M$, hence $f$ and $g$ are also neighboring in the conflict graph, and therefor cannot be processed simultaneously.

   \begin{figure}[!ht]
        \includegraphics[page=1]{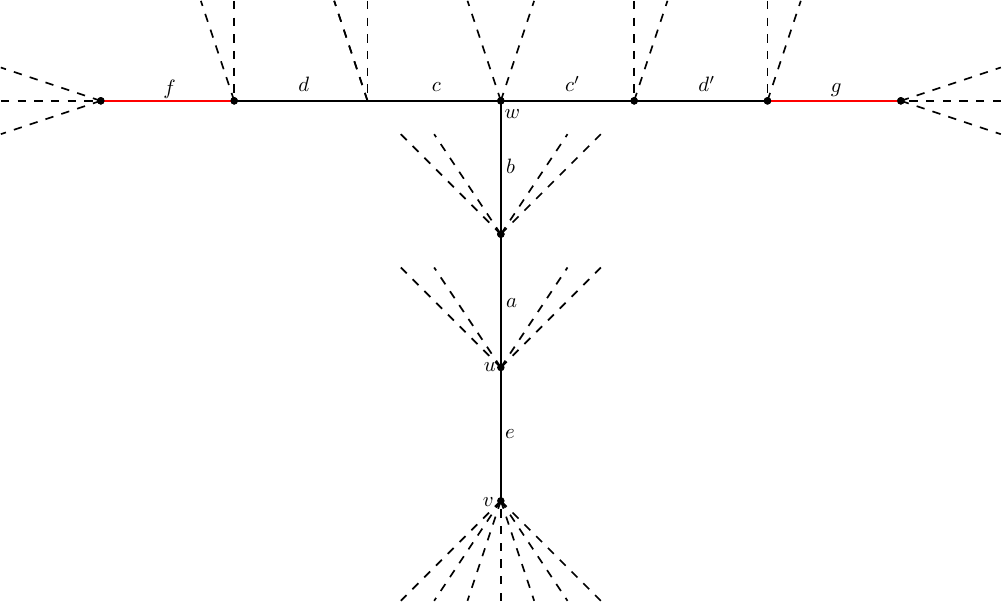}
        \centering
        \caption{The edges $f$ and $g$ both have distance $5$ to $e$ and are neighboring in the conflict graph. They share the edges $a$ and $b$ on their shortest path to $e$. We can conclude that $f$ and $g$ are also neighboring in the conflict graph.}
        \label{fig:CONGEST_disjoint_paths}
    \end{figure}
   
    This means that there can be at most $\Delta'^2$ active edges for any edge $e$. Since each edge needs to send $O(\Delta')$ messages, this is $O(\Delta'^3)$ messages in total. Each such message consists of $O(\log \Delta' )$ bits, as described above. The bottleneck in a message size is now the vertex IDs to label the message. However, we do not need unique vertex IDs beyond 5 hops. Hence we can use the $O(\Delta^6)$ vertex coloring as new vertex IDs of size $O(\log \Delta)$. This means that each message is $O(\log \Delta' )$ bits \emph{in total}. 

    Our total congestion now becomes $O(\Delta'^3 \log \Delta' )= O(\log^{1.5} n\log \log n)$ bits per edge, which can be transmitted in $O(\log^{0.5} n \log \log n)$ rounds, since we have a bandwidth of $\Theta(\log n)$ bits. Here we use that $\Delta'=\Theta(\sqrt{\log n})$. 

    So every instance of the coloring schedule takes $O(\log^{0.5} n \log \log n)$ rounds, hence we obtain total round complexity of $O(\Delta'^4)\cdot O(\log^{0.5} n \log \log n)= \tilde O(\log^{2.5} n)$ rounds for executing the schedule.  
\end{proof}

\subsection{ \texorpdfstring{$(2\Delta-1)$}{2Delta-1}-Edge Coloring}
\label{sec:congest_cor}

We show that the following corollary follows immediately from \Cref{thm:CONGESTmain}. 

\CorTwoDelta*
\begin{proof}
    \Cref{thm:CONGESTmain} gives the result for graphs with $\Delta\ge \Theta(\sqrt{\log n})$. When $\Delta= O(\sqrt{\log n})$, we can use the deterministic algorithm from Barenboim, Elkin, and Goldenberg~\cite{BEG17} that edge colors a graph with $2\Delta-1$ colors in $O(\Delta+\log^*n)$ rounds in \CONGEST. For $\Delta=\Theta(\sqrt{\log n})$, this reduces to $O(\sqrt{\log n}) = \tilde O(\log^{2.5}n+\log^2 \Delta \log n)$. 
\end{proof}

\appendix
\newpage
\section{Further Related Work}\label{sec:related work}

\paragraph{Edge coloring.}
There has been extensive research on various aspects of edge coloring in distributed models of computation. We also refer the reader to the detailed related work section in \cite{JMS25} and to the excellent book by Barenboim and Elkin on early results on distributed graph coloring \cite{barenboimelkin_book}.

In \Cref{sec:intro}, we have already discussed the relevant deterministic work on edge coloring with fewer than $2\Delta-1$ colors. There are two results that we have not mentioned. First, there is an $(3/2+\eps)\Delta$-edge coloring algorithm in $\tilde O(\log^2\Delta\log n)$ rounds \cite{BMNSU25}. Second, \cite{JMS25} give a faster algorithm (for general $\Delta)$ that runs  in $\tilde O(\log^{5/3}n)$ rounds, but it uses $2\Delta-2\gg (1+\eps)\Delta$ colors and hence only saves a single color compared to the greedy version of the problem.

For coloring with $2\Delta-1$ colors a series of work \cite{K20,BKO20,BBKO22} culminated with a $\poly\log \Delta+O(\log^*n)$ algorithm in the \LOCAL model and with a similar complexity but $(8+\eps)\Delta$ colors in the \CONGEST model. As reasoned,  algorithms with such an efficient runtime cannot exist for coloring with fewer colors even when restricted to graphs with  degrees being at least $\Omega(\sqrt{\log n})$.  When parameterized solely as a function of $n$ the first $\poly\log n$-round deterministic algorithm for greedy edge coloring was published by Fischer, Ghaffari, and Kuhn in 2017, even before one could efficiently compute network decompositions~\cite{FGK17}. The work was later improved to $\tilde O(\log^2\Delta\log n)$ rounds by Harris \cite{H20}. 

As mentioned in \Cref{sec:intro}, there are several randomized  algorithms usually targeting $(1+\eps)\Delta$ colors, starting with the relatively slow R\"odl Nibble method-based algorithm from Panconesi, Dubashi, and Grable \cite{NibbleMethod} and a  $\Delta+O(\log n/\log\log n)$-edge coloring algorithm  by Su and Vu in $\poly(\Delta,\log n)$ rounds~\cite{SuVu19},   over a $\poly\log\log n$-round algorithm in \LOCAL \cite{Davies23} and \CONGEST \cite{HMN22}, even to a $O(\log^*n)$-round algorithm for large degree graphs \cite{CHLPU18,HN21}.  Randomized, the problem has an $\Omega(\log_{\Delta}\log n)$ lower bound and it was attained on trees in the same paper \cite{CHLPU18}.

There are only a few works that previously explicitly aimed at designing \emph{deterministic} LOCAL algorithms. 
\cite{GKMU18} provides a $(1+\eps)\Delta+O(\log n)$-edge coloring in $O(\log^7 n)$ rounds, which remains inherently at least $\Omega(\log^3 n)$, even if one  plugs in significantly improved algorithms for subroutines such as computing hypergraph maximal matchings \cite{GG24}. 

Another technically involved result  by Bernshteyn and Dhawan \cite{Bernshteyn, VizingBoundedDegree}  gave a deterministic $(\Delta + 1)$-edge coloring algorithm in $\widetilde{\mathcal{O}}(\Delta^{84} \cdot \log^5 n)$ rounds. While it uses fewer colors, it is significantly slower and relies on deep and intricate techniques from descriptive combinatorics. 
Alternatively, one can obtain $\poly\log n$-round deterministic algorithms for the problem, by taking the state-of-the-art randomized edge coloring algorithms, e.g.,  \cite{Davies23,HMN22,HN21},  and leveraging them by the extremely powerful derandomization framework developed in a series of papers \cite{GhaffariKM17,GHK18,RG20,GG24}. 

\paragraph{Connections between online, sequential local and distributed models.}
A prior effort to bridge between online, sequential local and distributed models has been done in \cite{akbari-eslami-etal-2023-locality-in-online-dynamic}. They define the so-called Online-LOCAL model. Like classic online algorithms, it has unbounded centralized memory, but it also allows access to the $T$-hop neighborhood (in the to-be-revealed graph) of any vertex or edge that has already arrived. This can be viewed as an extension of the \SLOCAL model revelant to our work (see \Cref{sec:technical_overview}) with centralized memory.   See  \cite{akbari-coiteux-roy-etal-2025-online-locality-meets} for the definition of a randomized version of the Online-LOCAL model. Their main result shows that the unbounded centralized memory in comparison to the \SLOCAL model does not help significantly for  solving locally checkable graph problems on paths, cycles, and rooted trees.

\paragraph{Online coloring.}
Sleater and Tarjan~\cite{SleatorT85} initiated the research in online algorithms. Since then, many graph problems have been studied, such as matching~\cite{KarpVV90}, independent sets~\cite{HalldorssonIMT02}, vertex cover~\cite{DemangeP05}, dominating set~\cite{BoyarEFKL19}, and, of course, online edge coloring. 
As mentioned before, a simple greedy algorithm that for every arriving edge picks an available color achieves a $2\Delta-1$ coloring. It was already shown over thirty years ago~\cite{BMN92} that using a deterministic algorithm, one cannot improve on this -- even by a single color -- for graphs with maximum degree $\Delta=O(\log n)$. Similarly, one cannot do better with a randomized algorithm for graphs with maximum degree $\Delta=O(\sqrt{\log n})$. For arbitrary $\Delta$, there is also a lower bound saying that $\Delta + \Omega(\sqrt{\Delta})$ colors are needed \cite{CohenPW19}.

Since then, many have obtained results in restricted settings, such as random order edge-arrivals~\cite{AggarwalMSZ03,BahmaniMM12,BhattacharyaGW21,KulkarniLSST24,DudejaG025}, bipartite graphs with one sided vertex-arrivals~\cite{CohenPW19,BlikstadSVW24bipartite}, vertex arrivals~\cite{SaberiW21}, and initial results for the general case of edge-arrivals~\cite{KulkarniLSST24,BlikstadSVW24edge}.

Very recently, \cite{BlikstadSVW25} gave the algorithm that we build on in this work. It provides near-matching upper bounds: they give a deterministic algorithm that computes a $(1+o(1))\Delta$ edge coloring for graphs with $\Delta \ge \omega(\log n)$, and they give a randomized algorithm that computes a $(1+o(1))\Delta$ edge coloring for graphs with $\Delta \ge \omega(\sqrt{\log n})$, against an oblivious adversary, matching the thresholds in the lower bounds from \cite{BMN92}.


 \newpage
\section{Degree Splitting in the \CONGEST Model}\label{sec:Congest_DS}
As seen in \Cref{thm:LOCALmainNoND}, we  use degree splitting to reduce the general case to a smaller maximum degree. 
We already used a degree splitting lemma in the \LOCAL model, \Cref{lm:degree_splitting_LOCAL}. In the \CONGEST model, no such result can be found in the literature. However, with minor adaptations, we can use the \LOCAL algorithm. 

Basically, \cite{GhaffariHKMSU20} works in the \CONGEST model. The only subroutine in there that is \LOCAL-specific is sinkless orientation. Recent work~\cite{Nolin25} shows that sinkless orientation can also be done in $ O(\log n)$ rounds in the \CONGEST model, see \Cref{lm:SO_CONGEST}.

\paragraph{Sinkless orientation.}
As a subroutine in \Cref{lm:degree_spliting_CONGEST}, we need to solve the sinkless orientation problem. Let us first define the problem. 

\begin{definition}
    Given a graph where each edge is given an orientation (i.e., for an edge $uv$, either the edge is oriented from $u$ to $v$ or from $v$ to $u$), a \emph{sink} is a node such that all of its incident edges are oriented towards it. 
    
    In the \emph{Sinkless Orientation} problem, the goal is to orient all edges such that every node of degree $3$ or higher is not a sink, i.e., has at least one outgoing edge.
\end{definition}

Nolin~\cite{Nolin25} gave a deterministic $O(\log n)$ \CONGEST algorithm for the problem. 

\begin{lemma}[Sinkless Orientation~\cite{Nolin25}]\label{lm:SO_CONGEST}
    There is a deterministic \CONGEST algorithm that, given a graph with minimum degree $\delta$, computes a Sinkless Orientation in $O(\log_\delta n)$ rounds. 
\end{lemma}

\paragraph{Degree splitting.}
For self-containedness of this section, let us restate the problem in degree splitting. 
The undirected degree splitting problem seeks a partitioning of the graph edges $E$ into two parts so that the partition is nearly balanced around each node. Concretely, we should color each edge red or blue such that for each node, the difference between its number of red and blue edges is at most some small discrepancy value $\kappa$. In other words, we want an assignment $q \colon E \to \{+1, -1\}$ such that for each node $v\in V$, we have
\begin{equation*}
    \left|\sum_{e\in E(v)}q(e)\right| \le \kappa,
\end{equation*}
where $E(v)$ denotes the edges incident on $v$. We want $\kappa$ to be as small as possible.

In the \emph{directed} version, we orient all edges such that, at each node, the difference between its number of incoming and outgoing edges is at most the discrepancy value $\kappa$.

We can now state the \CONGEST degree splitting result and explain how to obtain it from the algorithm in \cite{GhaffariHKMSU20}.
\begin{lemma}[\CONGEST Degree Splitting]\label{lm:degree_spliting_CONGEST}
For every $\eta> 0$, there is a deterministic $\tilde O(\eta^{-1} )T_{SO}$
round \CONGEST algorithm for computing (directed) degree splitting with discrepancy at each node at most $\eta \cdot d(v) +O(1)$, where $T_{SO}$ is the round complexity of solving the sinkless orientation problem. 
\end{lemma}
\begin{proof}[Proof Sketch]
Before we reason on how to obtain the so-called undirected splitting result of the lemma, we explain how \cite{GhaffariHKMSU20} obtains a directed splitting. The objective in the directed splitting is to orient the edges of the input graph such that the discrepancy between the number of incoming and outgoing edges at each node is minimized. 

At the core the degree splitting algorithm of \cite{GhaffariHKMSU20}, arbitrarily splits every node $v$ with degree $d(v)$ into $\approx d(v)/3$ virtual nodes  of degree $3$ each. Then, it computes a sinkless orientation of the resulting graph providing node $v$ with $d(v)/3$ outgoing edges. Now, $v$ can exclusively decide on the fate of all these outgoing edges without disturbance by other nodes. So, $v$ can re-orient these edges such that half of them are oriented outwards and half of them are oriented inwards. As a result this would limit the discrepancy between the number of incoming and outgoing edges at $v$ to the $2d(v)/3$ incident edges that it does not control. To obtain a better split, the algorithm now performs a clever trick. Each node pairs its $d(v)/3$ outgoing edges arbitrarily. Let $e=\{v,u\}$ and $f=\{v,w\}$ be the edges in one such pair. Now, node $v$ replaces these edges with a virtual edge $\{u,w\}$ corresponding to a length $2$ paths in the original communication network and effectively erasing the original edges. Let $H$ be the resulting (multi-)graph. After doing so, node $v$ is left with $\approx 2d(v)/3$ incident edges, where each of these edges could be the start/end of a path of length $2$ in the original network or the start/end of a normal edge. Now, to get a better split, they recurse on $H$ to compute a better orientation of these (possibly virtual) edges. The benefit is that after orienting these virtual edges aka paths, we can infer a natural orientation of the edges of the original graph. So, if the virtual edge $\{u,w\}$ was oriented from $u$ to $w$, we orient $e$ from $u$ to $v$ and $f$ from $v$ to $w$; a similar procedure is done for longer paths.  This provides a perfect split at $v$ between incoming and outgoing edges whenever $v$ is not an endpoint of a path. As the number of incident paths reduces by a factor $2/3$ with each recursion level, after $\log \Delta$ recursion levels we would have achieved a split with constant discrepancy.

Bounding the runtime is non-trivial for the following reason. One requires $2$ rounds of communication in the original communication network to perform one round of communication over a virtual edge representing a path of length $2$, and more generally $k$ rounds for a path of length $k$. As the paths length increases exponentially with the recursion depths, one has to carefully limit it. E.g., $O(\log \Delta)$ recursion levels would immediately lead to a runtime of $2^{O(\log \Delta)}=\poly\Delta$.

Thus, the algorithm consists of several levels of recursion and it also performs some bootstrapping which replaces the original sinkless orientation algorithm with a splitting algorithm obtained by applying the explain algorithm only with few levels of recursion. Carefully balancing the recursion levels provides and the bootstrapping breakpoints provides the desired runtime and can bring the discrepancy of the directed split down to $\eta d(v) +O(1)$. The details are spelled out in great detail in \cite{GhaffariHKMSU20}.

\paragraph{\boldmath \CONGEST implementation.}
All virtual edges represent edge disjoint paths in the original communication network. Vertices can locally keep track on which edges they merged. Hence we can also communicate a \CONGEST message in $k$ rounds over a length $k$ path by simple forwarding the messages to the respective recipient. This also works if the underlying network can only send bandwidth restricted messages. 
Thus, the whole procedure works in the \CONGEST model, if the underlying sinkless orientation algorithm does. We obtain a directed degree splitting result by using the algorithm from \Cref{lm:SO_CONGEST}. 

\paragraph{Undirected splitting.} To obtain the same result for undirected degree splitting, one has to undo long paths differently, by coloring its internal edges alternatingly red and blue. This provides a perfect split at every node that is not at the endpoint of a path. Hence the discrepancy is limited to the number of incident virtual paths in the lowest recursion level. Clearly coloring these paths alternatingly can also be implemented in the \CONGEST model in time proportional to the length of the path and with no overhead compared to the \LOCAL model algorithm. For further details on the algorithm and its correctness we refer to \cite{GhaffariHKMSU20}.
\end{proof}


\newpage
  \printbibliography[heading=bibintoc]

\end{document}